\documentclass[orivec]{llncs}

\usepackage{MnSymbol}
\usepackage{graphicx}
\usepackage{enumitem}
\usepackage{longtable}
\usepackage{booktabs} 
\usepackage{amsfonts}%
\usepackage{subcaption}
\usepackage{algorithm}
\usepackage[noend]{algpseudocode}
\usepackage[noadjust]{cite}
\usepackage[utf8]{inputenc}
\usepackage[T1]{fontenc}
\usepackage{csquotes}
\usepackage{sidecap}
\usepackage{todonotes}
\usepackage{multicol}
\usepackage{authblk}

\algrenewcommand\algorithmicindent{0.35em}%

\newcommand\skeygc{\textsc{Kre-Ygc}}
\newcommand\skeahel{\textsc{Kre-Ahe1}}
\newcommand\skeahed{\textsc{Kre-Ahe2}}
\newcommand\skeshe{\textsc{Kre-She}}

\newcommand\ygc{\textsc{Kre-Ygc}}
\newcommand\ahel{\textsc{Kre-Ahe1}}
\newcommand\ahed{\textsc{Kre-Ahe2}}
\newcommand\she{\textsc{Kre-She}}

\newcommand\skshare{sks}
\newcommand\partdec{Decp}
\newcommand\findec{Decf}
\newcommand\dgkdec{\textsc{DgkDec}}
\newcommand\dgkeva{\textsc{DgkEval}}
\newcommand\paired{\textsc{Paired}}
\newcommand\dgkcompare{\textsc{DgkCompare}}

\newcommand\kre{\mathrm{KRE}}
\newcommand\simulator{\mathrm{SIM}}

\newcommand\ctxtcmt{\alpha} 
\newcommand\ctxtlen{\kappa} 
\newcommand\secprm{\lambda}
\newcommand\inputlen{\mu}
\newcommand\bb{a} 

\newcommand{\func}[2]{\mathcal{#1}_{\mathsf{#2}}}  
\newcommand{\view}[2]{\mathsf{View}_{\mathsf{#1}}^{#2}}

\newcommand\krsym{$k\textsuperscript{th}$-ranked}
\newcommand\krsymt{$k\textsuperscript{th}$-Ranked}

\newcommand{\bitrep}[1]{#1^{\mathsf{b}}}
\newcommand{\ctxtrep}[1]{\lsem #1 \rsem}

\newcommand\sheadd{\textsc{SheAdd}} 
\newcommand\shemul{\textsc{SheMult}} 
\newcommand\shecmp{\textsc{SheCmp}} 
\newcommand\shefadder{\textsc{SheFadder}} 
\newcommand\sheequ{\textsc{SheEqual}} 

\algnewcommand\algorithmicforeach{\textbf{for each}}
\algdef{S}[FOR]{ForEach}[1]{\algorithmicforeach\ #1\ \algorithmicdo}

\begin{document}

\title{Secure Computation of the \krsymt~Element in a Star Network}











\author{Anselme Tueno\inst{1} \and
Florian Kerschbaum\inst{2} \and
Stefan Katzenbeisser\inst{3} \and
Yordan Boev \inst{4} \and
Mubashir Qureshi \inst{5}}

\institute{SAP SE, Germany
\email{anselme.kemgne.tueno@sap.com} 
\and
University of Waterloo, Canada
\email{florian.kerschbaum@uwaterloo.ca} 
\and
University of Passau, Germany
\email{Stefan.Katzenbeisser@uni-passau.de}
\and
SAP SE, Germany
\email{iordan.boev@gmail.com} 
\and 
SAP SE, Germany
\email{mubashir.mehmood.qureshi01@sap.com} 
}

\maketitle

\begin{abstract}
We consider the problem of securely computing the \krsym~element in a sequence of $n$ private integers distributed among $n$ parties. 
The \krsym~element (e.g., minimum, maximum, median) is of particular interest in benchmarking, which allows a company to compare its own
key performance indicator to the statistics of its peer group. The individual integers are sensitive data, yet the \krsym~element is of mutual interest to the parties. Previous secure computation protocols for the \krsym~element require a communication channel between each pair of parties.  
They do not scale to a large number of parties as they are highly interactive resulting in longer delays. Moreover, they are difficult to deploy as special arrangements are required between each pair of parties to establish a secure connection. A server model naturally fits with the client-server architecture of Internet applications in which clients are connected to the server and not to other clients. It can simplify  secure computation by reducing the number of rounds, and as a result, improve its performance and scalability.
In this model, there are communication channels only between each client and the server, while only clients provide inputs to the computation. Hence, it is a centralized communication pattern, i.e., a star network. We propose different approaches for privately computing the \krsym~element in the server model, using either garbled circuits or threshold homomorphic encryption. Our schemes have a constant number of rounds and can compute the \krsym~element within seconds for up to 50 clients in a WAN. 
\end{abstract}


\section{Introduction}
Given $n$ parties each holding a private integer, we consider the problem of securely computing the \krsym~element (KRE) of these $n$ integers. This is a secure multiparty computation (SMC) where several parties wish to compute a publicly known function on their private input while revealing only the output of the computation to a designated subset of parties and nothing else. The computation of the \krsym~element is of particular interest in settings such as collaborative benchmarking, where the individual inputs are sensitive data, yet  the \krsym~element is of mutual interest to all parties \cite{AggarwalMP.2010, Kerschbaum.2008J}. 

\paragraph{\textbf{Benchmarking.}}
A key performance indicator (KPI) is a statistical quantity measuring the performance of a business process.
Benchmarking is a management process where a company compares its KPI to the statistics of the same KPIs of a group of competitors from a peer group. 
Examples of KPIs from different company operations are make cycle time (manufacturing), cash flow (financial) and employee fluctuation rate (human resources).
A peer group is a group of similar companies, usually competitors, wanting to compare against each other. Examples formed along different characteristics include car manufacturers (industry sector), Fortune 500 companies in the United States (revenue and location). 
A big challenge for benchmarking is that KPIs are sensitive and confidential, even within a single company~\cite{Kerschbaum.2008J}. 

\paragraph{\textbf{Confidentiality.}} 
Confidentiality is of the utmost importance in benchmarking, since KPIs
allow the inference of sensitive information. Companies are therefore hesitant to share their business performance data due to the risk of losing a competitive advantage \cite{Kerschbaum.2008J}. The confidentiality issue can be addressed using SMC~\cite{Yao.1982, GoldreichMW.1987, Ben-OrGW.1988}, which guarantees that no party will learn more than the output of the protocol, i.e., the other parties’ inputs remain confidential. 

\paragraph{\textbf{Communication Model.}}
There exist several secure protocols that can be used for keeping KPIs confidential while comparing them~\cite{Yao.1982, GoldreichMW.1987, Ben-OrGW.1988, AggarwalMP.2010}. They require
a communication channel between each pair of input parties. We will refer to this approach as the \emph{standard model}. Protocols in the standard model do not scale easily to a large number of parties as they require a communication channel between any pair of parties and are highly interactive, resulting in high latency. Moreover, they are difficult to deploy as special arrangements are required between each pair of parties to establish a secure connection~\cite{CatrinaK2008}. 
A promising approach for overcoming these limitations is to use the help of a small set of untrusted non-colluding servers. We will therefore refer to it as the \emph{server model}. In this model, the servers make their computational resources available for the computation, but have no input to the computation and receive no output \cite{Kerschbaum.2008J, KamaraMR2012}. For example, Jakobsen et al. \cite{JakobsenNO2014} propose a framework in which the input parties (the clients) delegate the computation to a set of untrusted workers.
Relying on multiple non-colluding servers requires a different business model for the service provider of a privacy-preserving service. The service provider has to share benefits with an almost equal peer offering its computational power \cite{Kerschbaum.2009}.
We therefore use a communication model consisting of clients (with private inputs) and a server. In this model, the server provides no input to the computation and does not learn the output, but makes its computational resources available to the clients \cite{Kerschbaum.2008J, KamaraMR2012}. 
Moreover, there are communication channels only between each client and the server. Hence, it is a centralized communication pattern, i.e., a star network.
As a result, the clients will only communicate with the server, but never directly amongst each other.
This model naturally fits with the client-server architecture of Internet applications and allows a service provider to play the server's role. It can simplify the secure  protocol, and improve its performance and scalability~\cite{CatrinaK2008, KamaraMR2012, KamaraMR11}. 
 \begin{table}[tbp]
	\centering
	\begin{tabular}{|c|l|l|}
	\hline
	\multicolumn{1}{|c|}{\scriptsize{Symbol}}  & {\scriptsize{Interpretation}}	\\  \hline
	\multicolumn{1}{|c|}{\scriptsize{$\inputlen$}}  & {\scriptsize{Bitlength of inputs}} \\ 
	\multicolumn{1}{|c|}{\scriptsize{$n$}}  & {\scriptsize{Number of clients}} \\ 
	\multicolumn{1}{|c|}{\scriptsize{$t$}}  & {\scriptsize{Secret sharing threshold,} $t\leq n$} \\ 
	\multicolumn{1}{|c|}{\scriptsize{$\ctxtlen$}}  & {\scriptsize{Bitlength of asymmetric ciphertext}} \\ 
	\multicolumn{1}{|c|}{\scriptsize{$\secprm$}}  & {\scriptsize{Security parameter}} \\ 
	\multicolumn{1}{|c|}{\scriptsize{$x_1, \ldots, x_{n}$}}  & {\scriptsize{Client's inputs}} \\ 
	\multicolumn{1}{|c|}{\scriptsize{$x_{i}^{b} = x_{i\inputlen} \ldots x_{i1}$}}  & {\scriptsize{Bit representation of $x_i$ with most significant bit (MSB) $x_{i\inputlen}$ }} \\ 
	\multicolumn{1}{|c|}{\scriptsize{$|y|$}}  & {\scriptsize{Bitlength of integer $y$, e.g., $|x_i| = \inputlen$}} \\ 
	\multicolumn{1}{|c|}{\scriptsize{$\lsem x_i \rsem$}}  & {\scriptsize{$x_i$'s ciphertext under public key $pk$} } \\ 
	\multicolumn{1}{|c|}{\scriptsize{$\lsem x_i \rsem_j$}}  & {\scriptsize{$x_i$'s ciphertext under public key $pk_j$} } \\ 
	\multicolumn{1}{|c|}{\scriptsize{$\lsem x_{i}^{b} \rsem$}}  & {\scriptsize{Bitwise encryption $(\lsem x_{i\inputlen}\rsem, \ldots, \lsem x_{i1} \rsem)$}} \\
	\multicolumn{1}{|c|}{\scriptsize{$i \stackrel{\$}{\leftarrow} \mathbb{S}$}}  & {\scriptsize{Choose a random element $i$ in set $\mathbb{S}$}} \\
	\multicolumn{1}{|c|}{\scriptsize{$\{i_1, \ldots, i_t\} \stackrel{\$}{\leftarrow} \mathbb{S}$}}  & {\scriptsize{Choose $t$ random distinct elements in $\mathbb{S}$}} \\
	\multicolumn{1}{|c|}{\scriptsize{$\mathfrak{S}_n$}}  & {\scriptsize{Set of all permutations of $\{1, \ldots, n\}$}} \\
	\hline
	\end{tabular}
	\captionsetup{justification=centering}
	\caption{Notations.}
	\label{Notation_Table}
 \end{table}
 \begin{table}[tbp]
 	\centering
	\begin{tabular}{| l | l | l | l | l | l |}
		\hline
			\multicolumn{1}{ |l|}{}                    & \scriptsize{\ygc}    &	\scriptsize{\ahel} &	\scriptsize{\she} & \scriptsize{\cite{AggarwalMP.2004}} \\ 
													& &    	\scriptsize{\ahed} &  &\\ \hline
		\multicolumn{1}{ |l|}{\scriptsize{\# Rounds}}	           & \scriptsize{4} & \scriptsize{4} & \scriptsize{2} & \scriptsize{$O(\mu)$}\\ 
		\multicolumn{1}{ |l|}{\scriptsize{Collusion-resis.}} & \scriptsize{$n-1 ~|~ 0$} &	\scriptsize{$t-1 ~|~ t$} & \scriptsize{$t-1 ~|~ t$} & \scriptsize{$t-1 ~|~$n/a}\\ 
		\multicolumn{1}{ |l|}{\scriptsize{Fault-tolerance}}	     & \scriptsize{0} & 	\scriptsize{$n-t$} & \scriptsize{$n-t$} & \scriptsize{$n-t$}\\ \hline
	\end{tabular}
	\caption{Schemes' properties: \textnormal{The collusion row refers to the number of parties that can collude - server excluded | server included - without breaking the privacy of non-colluding clients. The fault-tolerance row refers to the number of parties that can fail without preventing the protocol to properly compute the intended functionality.}}
	\label{KRE_Comparison}
 \end{table}
  \begin{table}[tbp]
	\centering
	\begin{tabular}{| l | l | l | l | l | l | l |}
		\hline
		\multicolumn{1}{ |l|}{}                  & \multicolumn{2}{ l|}{ \scriptsize{\ygc}} & \scriptsize{\ahel} &	\scriptsize{\she} & \scriptsize{\cite{AggarwalMP.2004}} \\ \cline{2-3}
		\multicolumn{1}{|l|} {}                      &  \scriptsize{sym.}      &      \scriptsize{asym.}               & \scriptsize{\ahed}     &    &  \\ \cline{1-6} 
		\multicolumn{1}{ |l|}{\scriptsize{CC-C}} & \scriptsize{$O(n\inputlen)$} & \scriptsize{$O(n)$}           & \scriptsize{$O(n\inputlen)$}   & \scriptsize{$O(\inputlen)$}  & \scriptsize{$O(n\inputlen^2)$}        \\ 
		\multicolumn{1}{ |l|}{\scriptsize{CC-S}} & \scriptsize{$O(n^2\inputlen)$} & \scriptsize{$O(n\log n)$} & \scriptsize{$O(n^2\inputlen)$}   & \scriptsize{$O(n^2\inputlen \log \inputlen)$}  & \scriptsize{n/a}         \\ 
		\multicolumn{1}{ |l|}{\scriptsize{BC-C}} & \scriptsize{$O(n\inputlen\secprm) $} & \scriptsize{$O(n\ctxtlen)$} & \scriptsize{$O(n\inputlen\ctxtlen)$}   & \scriptsize{$O((\inputlen + n)\ctxtlen)$} & \scriptsize{$O(n\inputlen^2\secprm)$}          \\ 
		\multicolumn{1}{ |l|}{\scriptsize{BC-S}} & \scriptsize{0} &\scriptsize{$O(n^2\ctxtlen)$}              & \scriptsize{$O(n^2\inputlen\ctxtlen)$}          & \scriptsize{$O(n\ctxtlen)$} & \scriptsize{n/a}  \\ \hline
	\end{tabular}
	\caption{Schemes' Complexity: \textnormal{Rows CC-C/S and BC-C/S denote the computation and communication (bit) complexity for 
	each client and the server, respectively. The columns \enquote{sym.} and \enquote{asym.} denote symmetric and asymmetric operations in \ygc.}}
	\label{KRE_Comparison_Compl}
 \end{table}
%
\paragraph{\textbf{Contribution.}} 
In summary, we propose different approaches for securely computing the \krsym~element (KRE) in a star network using either garbled circuits (GC) or additive homomorphic encryption (AHE) or somewhat homomorphic encryption (SHE): 
\begin{itemize}
	\item Our first scheme \skeygc~ uses Yao's GC \cite{LindellP.2009, BellareHR.2012} to compare clients' inputs.
	\item Our second scheme \skeahel~ is based on threshold AHE. We first propose a modified variant of the Lin-Tzeng comparison protocol \cite{LinT05}. 
	The server then uses it to compare inputs encrypted with AHE.
	\item In our third scheme \skeahed, we continue with threshold AHE, however, we perform the comparison using the DGK protocol \cite{DamgardGK.2007}.
	\item The fourth scheme \skeshe~ is based on SHE and allows the server to non-interactively compute the KRE such that the clients only interact to jointly decrypt the result.
\end{itemize}

\noindent We compare  the approaches in Tables \ref{KRE_Comparison} and \ref{KRE_Comparison_Compl} using the following measures:
\begin{itemize}
	\item Number of rounds: In contrast to \cite{AggarwalMP.2004}, all our protocols have a constant number of rounds.  
	\item Collusion-resistance: This is a protocol property that is measured by the number of parties that can collude without violating the privacy of the non-colluding ones. In \skeygc~ a collusion with the server completely breaks the security, while \skeahel~ and \skeahed~ can tolerate the collusion of several clients with the server as long as the number of colluding clients is smaller than a threshold $t$. If the server does not collude, then \skeygc~ can tolerate up to $n-1$ colluding clients. Aggarwal et al.'s scheme \cite{AggarwalMP.2004} is collusion-resistant if implemented with a threshold scheme.
	\item Fault-tolerance: It is a protocol property that is measured by the number of parties that can fail without preventing the protocol to properly compute the intended functionality. Our server model can only tolerate clients' failure. \skeygc~ is not fault-tolerant while \skeahel~ and \skeahed~ can tolerate failure of up to $n-t$ clients. Aggarwal et al.'s scheme \cite{AggarwalMP.2004} is fault-tolerant if implemented with a threshold scheme.
	\item Complexity: This refers to the asymptotic computation complexity as well as the  communication complexity. A summary is illustrated in Table \ref{KRE_Comparison_Compl}. We provide a detailed analysis in Appendix \ref{complexity_analysis}.   
\end{itemize}

\paragraph{\textbf{Structure.}}
The remainder of the paper is structured as follows. We begin by presenting related work in Section \ref{Related_work} and some 
preliminaries in Section \ref{Preliminaries}. We present our security model in Section \ref{Security_Definition} and a technical overview in Section 
\ref{Technical_Overview}.  The different approaches are presented in Sections \ref{GC_Secure_KRE} to \ref{SHE_Secure_KRE}. We discuss some implementation details and evaluation results in Section \ref{Evaluation}, before concluding our work in Section \ref{Conclusion}.  We provide further details such as security proofs and complexity analysis in the appendix.

\section{Related Work}
\label{Related_work}

Our work is related to secure multiparty computation (SMC). There are several generic SMC protocols \cite{DamgardKLPSS.2013, KellerOS.2016, Ben-OrGW.1988, CramerDN2001} that can be used to compute the \krsym~element of the union of $n$ private datasets. In practice, specialized protocols offer better performance as they use domain knowledge of the required functionality to optimize the secure protocol.
Aggarwal et al.~\cite{AggarwalMP.2010} introduced the first specialized protocol for computing the \krsym~element. Their multiparty protocol performs a binary search in the input domain resulting in $O(\inputlen)$ comparisons and, hence, requiring $O(\inputlen)$ sequential rounds. Each round requires a secure computation that performs two summations with complexity $O(n\inputlen)$ and two comparisons with complexity $O(\inputlen)$. As a result each client requires $O(n\inputlen^2+\inputlen^2)$ operations and sends $O(n\inputlen^2+\inputlen^2)$ bits.
Our protocols perform $O(n^2)$ comparisons, that can be executed in parallel, and have a constant number of rounds. A summary of the complexity of our schemes is illustrated in Table \ref{KRE_Comparison_Compl}.

The server model for SMC was introduced in \cite{FeigeKN1994}. Kerschbaum \cite{Kerschbaum.2009} proposed an approach allowing a service provider to offer a SMC service by himself. The cryptographic study of the server model was initiated in \cite{KamaraMR2012, KamaraMR11}. They all provide a generic solution for SMC while our approaches propose specialized protocol for the \krsym~element. 
The computation of the \krsym~element is also addressed in \cite{BlassK18,BlassK19} where the server is replaced by a blockchain. While the first one \cite{BlassK18} relies on Fischlin's comparison protocol \cite{Fischlin.2001}, the second one \cite{BlassK19} relies on the DGK comparison protocol \cite{DamgardGK.2007}.  The technical difficulty relies in the fact that parties must prove correct execution of the protocol which is done using zero-knowledge proofs resulting in maliciously secure protocols. These protocol requires only 3 rounds of computation, however they leak the order of the inputs to the parties.

\section{Preliminaries}
\label{Preliminaries} 
A Garbled Circuit (GC) \cite{LindellP.2009, EjgenbergFLL.2012, BellareHR.2012, ZahurRE.2015} can be used to execute any function privately between two parties.
To evaluate a function $f$ on input $x_i, x_j$, a garbling scheme $(F, e) \leftarrow Gb(1^{\secprm}, s, f)$ takes a security parameter $\secprm$, a random seed $s$, a Boolean encoding of $f$ and outputs a garbled circuit $F$ and an encoding string $e$ that is used to derive corresponding garbled inputs $\bar{x}_i, \bar{x}_j$ from $x_i, x_j$, i.e. there is a function $En$ such that $\bar{x}_i \leftarrow En(e, x_i)$ and $\bar{x}_j \leftarrow En(e, x_j)$. The garbling scheme is correct if $F(\bar{x}_i, \bar{x}_j) = f(x_i, x_j)$.

A homomorphic encryption (HE) allows computations on ciphertexts by generating an encrypted result whose decryption matches the result of a function on the  plaintexts. A HE scheme consists of the following algorithms:

\begin{itemize}
\item $pk, sk, ek \leftarrow KeyGen(\secprm)$: This probabilistic algorithm takes a security parameter $\secprm$ and outputs 
public, private and evaluation key $pk$, $sk$ and $ek$.
\item $c \leftarrow Enc(pk, m)$: This probabilistic algorithm takes $pk$ and a message $m$ and outputs a ciphertext $c$. We will denote $Enc(pk, m)$ by $ \lsem m \rsem$ (see Table \ref{Notation_Table}).
\item $c \leftarrow Eval(ek, f, c_1, \ldots, c_n)$: This probabilistic algorithm takes $ek$, an $n$-ary function $f$ and $n$ ciphertexts $c_1, \ldots c_n$ and outputs a ciphertext $c$.
\item $m' \leftarrow Dec(sk, c)$: This deterministic algorithm takes $sk$ and a ciphertext $c$ and outputs a message $m'$.
\end{itemize}

\noindent We require IND-CPA and the following correctness conditions $\forall m_1, \ldots, m_n$:
\begin{itemize}
	\item $Dec(sk, Enc(pk, m_i)) = Dec(sk, \lsem m_i \rsem) = m_i,$
	\item $Dec(sk, Eval(ek, f, \lsem m_1 \rsem, \ldots, \lsem m_n \rsem)) = Dec(sk, \lsem f(m_1, \ldots, m_n) \rsem)$.
\end{itemize}
If the scheme supports only addition, then it is \emph{additively homomorphic}. Schemes such as \cite{Paillier.1999, Koblitz1987}
are additively homomorphic and have the following properties:
\begin{itemize}
	\item Addition: $\forall m_1, m_2, \lsem m_1 \rsem \cdot \lsem m_2 \rsem = \lsem m_1 + m_2 \rsem$,
	\item Multiplication with plaintext: $\forall m_1, m_2, \lsem m_1 \rsem ^{ m_2} = \lsem m_1 \cdot m_2 \rsem$,
	\item Xor: $\forall a, b \in \{0,1\}, \textsc{Xor}(\lsem a \rsem, b) = \lsem a \oplus b \rsem = \lsem 1 \rsem^{b} \cdot \lsem a \rsem^{(-1)^b}$.
\end{itemize}

\noindent A Threshold Homomorphic Encryption (THE) \cite{BonehGGJKRS18, CramerDN2001} allows to share the private key to the parties using a threshold secret sharing scheme such that a subset of parties is required for decryption. Hence, instead of $sk$ as above, the key generation outputs a set of shares $\mathbb{SK} = \{\skshare_1, \ldots, \skshare_n\}$ which are distributed to the clients. The decryption algorithm is replaced by the following algorithms:
\begin{itemize}
\item $\tilde{m}_i \leftarrow \partdec(\skshare_i, c)$: The probabilistic partial decryption algorithm takes a ciphertext $c$ and a share $\skshare_i \in \mathbb{SK}$ of the private key and outputs $\tilde{m}_i$.
\item $m' \leftarrow \findec(\mathbb{M}_t)$: The deterministic final decryption algorithm takes a subset $\mathbb{M}_t =\{\tilde{m}_{j_1}, \ldots, \tilde{m}_{j_t}\} \subseteq \{\tilde{m}_1, \ldots, \tilde{m}_n \}$ of partial decryption shares and outputs a message $m'$.
\end{itemize}
We refer to it as \emph{threshold decryption}. It is correct if for all $\mathbb{M}_t =\{\tilde{m}_{j_1}, \ldots, \tilde{m}_{j_t}\}$ such that $|\mathbb{M}_t| \geq t$ and $\tilde{m}_{j_i} = \partdec(\skshare_{j_i}, \lsem m \rsem)$, it holds $m = \findec(\mathbb{M}_t)$.

\noindent When used in a protocol, we denote by \emph{combiner} the party which is responsible to execute algorithm $\findec()$. Depending on the protocol, the combiner can be any party. It receives a set $\mathbb{M}_t = \{\tilde{m}_{j_1}, \ldots, \tilde{m}_{j_t}\}$ of ciphertexts, runs $m' \leftarrow \findec(\mathbb{M}_t)$ and publishes the result or move to the next step of the protocol specification. 
 
\section{Security Definition}
\label{Security_Definition}
This section provides definitions related to our model and security requirements.
We start by defining the \krsym~element of a sequence of integers. 
\begin{definition}
\label{definition_KRE}
Let $\mathbb{X} = \{x_{1}, ..., x_{n}\}$ be a set of $n$ distinct integers and $\tilde{x}_1, \ldots, \tilde{x}_n$ be the corresponding sorted set, i.e., $\tilde{x}_1 \leq \ldots \leq \tilde{x}_n$, and $\mathbb{X} = \{\tilde{x}_1, \ldots, \tilde{x}_n\}$. The \emph{rank} of an element $x_i \in \mathbb{X}$ is $j$, such that $x_i = \tilde{x}_j$. The \emph{\krsym~element ($\kre$)} is the element $\tilde{x}_{ k }$ with rank $k$.
\end{definition}
If the rank is $k= \left\lceil  \frac{n}{2} \right\rceil$ then the element is called \textit{median}. If $k=1$ (resp. $k=n$) then the element is called \emph{minimum} (resp. \emph{maximum}).

\begin{definition}
Let $C_1, \ldots, C_n$ be $n$ clients each holding a private $\inputlen$-bit integer $x_1, \ldots, x_n$ and $S$ be a server which has no input. Our \emph{ideal functionality} $\func{F}{\kre}$ receives $x_1, \ldots, x_n$ from the clients, computes the $\kre$ $\tilde{x}_k$ and outputs $\tilde{x}_k$ to each client $C_i$. Moreover, $\func{F}{\kre}$ outputs a \emph{leakage} $\mathcal{L}_i$ to each $C_i$ and $\mathcal{L}_S$ to $S$.
\end{definition}

\noindent The leakage is specific to each protocol and contains information such as $n$, $t$, $\secprm$, $\ctxtlen$, $\inputlen$ (see Table \ref{Notation_Table}). It can be inferred from the party's view which is all that the party is allowed to learn from the protocol execution. In case of limited collusion (i.e., the number of colluding parties is smaller than a given threshold as given in Table \ref{KRE_Comparison}) additional leakage might include comparison results between some pair of inputs or the rank of some inputs.

\begin{definition}
\label{definition_view}
The \emph{view} of the $i$-th party during an execution of the protocol on input $\vec{x} = (x_1, \ldots, x_n)$ is denoted by:
$$ \view{i}{}(\vec{x}) = \{x_i, r_i, m_{i1}, m_{i2}, \ldots\},$$
where $r_i$ represents the outcome of the $i$-th party's internal coin tosses, and $m_{ij}$ represents the $j$-th message it has received.
\end{definition}
Since the server is a party without input, $x_{i}$ in its view will be replaced by the empty string.

We say that two distributions $\mathcal{D}_1$ and $\mathcal{D}_2$ are computationally indistinguishable (denoted $\mathcal{D}_1 \stackrel{c}{\equiv} \mathcal{D}_1$) if no probabilistic polynomial time algorithm can distinguish them except with negligible probability.

In this paper, we assume that parties follow the protocol specification, but the adversary keeps a record of all messages received by corrupted parties and tries to infer as much information as possible. Our adversary is, therefore, semi-honest. SMC security requires that what a party can learn from the protocol execution, can be inferred from its input and output only. The protocol is said secure if for each party, one can construct a simulator that given only the input and the output can generate a distribution that is computationally indistinguishable to the party's view.

\begin{definition}
\label{definition_SH_security}
Let 
$\func{F}{\kre}: (\{0,1\}^{\inputlen})^{n} \mapsto \{0,1\}^{\inputlen}$
be the functionality that takes $n$ $\inputlen$-bit inputs $x_1, \ldots, x_n$  and returns their KRE.  
Let $I = \{i_1, \ldots, i_t\} \subset \{1, \ldots, n+1\}$ be a subset of indexes of corrupted parties (Server's input $x_{n+1}$ is empty), $\vec{x} = (x_1, \ldots, x_n)$ and
$$\view{I}{}(\vec{x}) = (I, \view{i_1}{}(\vec{x}), \ldots, \view{i_t}{}(\vec{x})).$$
A protocol \emph{$t$-privately computes $\func{F}{\kre}$ in the semi-honest model} if there exists a polynomial-time simulator $\simulator$ such that: $\forall I, |I| = t$ and $\mathcal{L}_{I} = \bigcup_{i\in I} \mathcal{L}_i$, it holds:
$$\simulator(I, (x_{i_1}, \ldots, x_{i_t}), \func{F}{\kre}(x_1, \ldots, x_n), \mathcal{L}_{I}) \stackrel{c}{\equiv} \view{I}{}(x_1, \ldots, x_n).$$
\end{definition}

\section{Technical Overview}
\label{Technical_Overview}

In an initialization phase, clients generate and exchange necessary cryptographic keys through the server.
We assume the existence of a trusted third party (e.g., a certificate authority) which certifies public keys or generates keys for a threshold cryptosystem. Moreover the trusted third party is not allowed to take part in the main protocol or to collude with any party including the server.  
We stress that the initialization phase is run once and its complexity does not depend on the functionality that we want to compute. In the following, we therefore focus on the actual computations.

We determine the KRE in the main protocol by computing the rank of each $x_i$ and selecting the right one. To achieve that, we compare pairs of inputs $(x_i, x_j), 1 \leq i,j, \leq n$ and denote the result by a comparison bit $b_{ij}$. 

\begin{definition}
\label{comparison_bit_def}
Let $x_i, x_j, 1 \leq i,j, \leq n$, be integer inputs of $C_i, C_j$. Then the \emph{comparison bit} $b_{ij}$ of the pair $(x_i, x_j)$ is defined as 1 if $x_i \geq x_j$ and 0 otherwise.
The computation of $x_i \geq x_j$ is distributed and involves $C_i, C_j$, where they play different roles, e.g., generator and evaluator.
Similar to the functional programming notation of an ordered pair, 
we use \emph{head} and \emph{tail} to denote $C_i$ and $C_j$. 
\end{definition}

\noindent For each input $x_i$, we then add all bits $b_{ij}, 1\leq j \leq n$ to get its rank $r_i$.

\begin{lemma}
\label{rank_computation_lemma}
Let $x_1, \ldots, x_n$ be $n$ distinct integers, and let $r_1, \ldots, r_n \in\{1, \ldots, n\}$ be their corresponding ranks and $b_{ij}$ the comparison bit for $(x_i, x_j)$.
It holds $r_i = \sum_{j=1}^{n} b_{ij}.$
\end{lemma}
\begin{proof}
Since $r_i$ is the rank of $x_i$, $x_i$ is by definition larger or equal to $r_i$ elements in $\{x_1, \ldots, x_n\}$. This means that $r_i$ values among $b_{i1}, \ldots, b_{in}$ are 1 and the remaining $n-r_i$ values are 0. It follows that $\sum_{j=1}^{n} b_{ij} = r_i$.
\end{proof}
The above lemma requires distinct inputs. To make sure that clients' inputs are indeed distinct before the protocol execution, we borrow the idea of \cite{AggarwalMP.2010} and use the index of parties as differentiator. Each party $C_i$ represents its index $i$ as a $\log n$-bit string and appends it at the end (i.e., in the least significant positions) of the binary string of $x_i$, resulting in a new input of length $\inputlen+\log n$. 
For simplicity, we assume in the remainder of the paper, that the $x_i$'s are all distinct $\inputlen$-bit integers.
Therefore, it is not necessary to compare all pairs $(x_i, x_j), 1\leq i,j \leq n$, since we can deduce $b_{ji}$ from $b_{ij}$. 

As explained in Definition \ref{comparison_bit_def}, $C_i, C_j$ play different role in the comparison for $(x_i, x_j)$. Therefore, we would like to equally distribute the roles among the clients. As example for $n=3$, we need to compute only three (instead of nine) comparisons resulting in three \emph{head} roles and three \emph{tail} roles. Then we would like each of the three clients to play the role \emph{head} as well as \emph{tail} exactly one time. We will use Definition \ref{def:Paired} and Lemma \ref{Paired_Predicate_Lemma} to equally distribute the roles \emph{head} and \emph{tail} between clients.
\begin{definition}
	\label{def:Paired}
	Let $\mathbb{X} = \{x_1, \ldots, x_n\}$ be a set of $n$ integers. We define the predicate \paired~as follows:
	\begin{subequations}
\begin{align}
  \paired(i,j):=~ & (i  \equiv 1\;(\bmod\;2)  \land  i > j  \land  j  \equiv 1\;(\bmod\;2))~\vee \label{p1} \\
            & (i  \equiv 1\;(\bmod\;2)  \land  i < j  \land  j  \equiv 0\;(\bmod\;2))~\vee \label{p2} \\
	          & (i  \equiv 0\;(\bmod\;2)  \land  i > j  \land  j  \equiv 0\;(\bmod\;2))~\vee \label{p3} \\
	          & (i  \equiv 0\;(\bmod\;2)  \land  i < j  \land  j  \equiv 1\;(\bmod\;2)).     \label{p4} 
\end{align}
\label{eq:Paired_Predicate}
\end{subequations}
\end{definition}

\begin{lemma}
\label{Paired_Predicate_Lemma}
Let $\mathbb{X} = \{x_1, \ldots, x_n\}$ be a set of $n$ integers and the predicate $\paired$ be as above.	
Then comparing only pairs $(x_i, x_j)$ such that $\paired(i,j)=true$ is enough to compute the rank of all elements in $\mathbb{X}$.
\end{lemma}
\begin{proof}
Let $\mathbb{P} = \{(x_i, x_j): x_i, x_j \in \mathbb{X} \wedge i \neq j\}$, $\mathbb{P}_1 = \{(x_i, x_j): x_i, x_j \in \mathbb{X} \wedge \paired(i,j) = true\}$, $\mathbb{P}_2 = \{(x_i, x_j): x_i, x_j \in \mathbb{X} \wedge Q(i, j) = true\}$, where $Q(i,j)$ is defined as follows:

\begin{subequations}
\begin{align}
  Q(i,j):=~ & (i  \equiv 1\;(\bmod\;2)  \land  i < j  \land  j  \equiv 1\;(\bmod\;2))~\vee \label{q1} \\
            & (i  \equiv 0\;(\bmod\;2)  \land  i > j  \land  j  \equiv 1\;(\bmod\;2))~\vee \label{q2} \\
	          & (i  \equiv 0\;(\bmod\;2)  \land  i < j  \land j   \equiv 0\;(\bmod\;2))~\vee \label{q3} \\
	          & (i  \equiv 1\;(\bmod\;2)  \land  i > j  \land  j  \equiv 0\;(\bmod\;2)).      \label{q4}
\end{align}
\label{eq:Paired_Compl}
\end{subequations}

Clearly, $\mathbb{P}$ contains the maximum number of comparisons required to compute the rank of every $x_i \in \mathbb{X}$. Now it suffices to show that:
\begin{enumerate}
	\item \label{claim1} $\mathbb{P}_1$ and $\mathbb{P}_2$ form a partition of $\mathbb{P}$
	\item \label{claim2} $\forall ~(x_i, x_j) \in \mathbb{P}: (x_i, x_j) \in \mathbb{P}_1 \Leftrightarrow (x_j, x_i) \in \mathbb{P}_2$
\end{enumerate}
$\mathbb{P}_1$ and $\mathbb{P}_2$ are clearly subsets of $\mathbb{P}$. For each $(x_i, x_j) \in \mathbb{P}$, $(i,j)$ satisfies exactly one of the conditions (\ref{p1}), \ldots, (\ref{p4}), (\ref{q1}), \ldots, (\ref{q4}), hence $\mathbb{P} \subseteq \mathbb{P}_1 \cup \mathbb{P}_2$. Moreover, for each $(x_i, x_j) \in \mathbb{P}$, either $\paired(i, j) = true$ or $Q(i,j) = true$. It follows that $\mathbb{P}_1 \cap \mathbb{P}_2 = \emptyset$ which concludes the proof of claim \ref{claim1}. 
To prove claim \ref{claim2}, it suffices to see that, $(i,j)$ satisfies condition (\ref{p1}) if and only if $(j,i)$ satisfies condition (\ref{q1}). The same holds for (\ref{p2}) and (\ref{q2}), (\ref{p3}) and (\ref{q3}), (\ref{p4}) and (\ref{q4}).
\end{proof}

\noindent For example, if $n=3$, we compute comparison bits only for $(x_1, x_2)$, $(x_2, x_3)$, $(x_3, x_1)$ and deduce the remaining comparison bits from the computed ones. If $n=4$, we compare only $(x_1, x_2)$, $(x_1, x_4)$, $(x_2, x_3)$, $(x_3, x_1)$, $(x_3, x_4)$, $(x_4, x_2)$.

\noindent The  predicate \paired~(Equation \ref{eq:Paired_Predicate}) is used in our schemes to reduce the number of comparisons and to equally distribute the computation task of the comparisons among the clients.
Let $\#head_i$ (resp. $\#tail_i$) denote the number of times $\paired(i,j)=true$ (resp. $\paired(j,i)=true$) holds. For example, if $n=3$, we have $\#head_i = \#tail_i = 1$ for all clients. However, for $n=4$, we have $\#head_1=\#head_3 = 2$, $\#tail_1=\#tail_3 = 1$, $\#head_2=\#head_4 = 1$ and $\#tail_2=\#tail_4 = 2$.

\begin{lemma}
\label{Paired_CMP_Count_By_CLT}
Let $\mathbb{X} = \{x_1, \ldots, x_n\}$ be a set of integers and assume the predicate \paired~  is used to sort $\mathbb{X}$. If $n$ is odd then:
$ \#head_i = \#tail_i = \frac{n-1}{2}.$  
If $n$ is even then:
\[
\#head_i=
		\begin{cases}
		\frac{n}{2}     & \mbox{if} ~ i ~ \mbox{odd}\\  
		\frac{n}{2} - 1 & \mbox{if} ~ i ~ \mbox{even}
		\end{cases} 
~
\#tail_i=
		\begin{cases}
		\frac{n}{2} - 1 & \mbox{if} ~ i ~ \mbox{odd}\\  
		\frac{n}{2}     & \mbox{if} ~ i ~ \mbox{even}.
		\end{cases} 
\]
\end{lemma} 
\begin{proof}
This is actually a corollary of the proof of Lemma \ref{Paired_Predicate_Lemma}. It follows from the fact that $(x_i, x_j) \in \mathbb{P}_1 \Leftrightarrow (x_j, x_i) \in \mathbb{P}_2$ and any $x_i$ is involved in $n-1$ comparisons (since we need $b_{i1}, \ldots, b_{in}$ to compute $r_{i} = \sum_{j=1}^{n}b_{ij}$, where we trivially have $b_{ii} = 1$ without comparison). This proves the case when $n$ is odd. If $n$ is even then the odd case applies for $n' = n-1$. Then for each $i \in \{1, \ldots, n'\}$, we have $\paired(i, n) = true$ if $i$ is odd (condition \ref{p2}) and $\paired(n, i) = true$ if $i$ is even (condition \ref{p3}).
\end{proof}

\section{Protocol \skeygc}
\label{GC_Secure_KRE}
This section describes  \skeygc~ (Protocol \ref{YGC_KRE_Protocol}) based on GC which consists of an initialization and a main protocol. During initialization, parties generate and distribute cryptographic keys. The online protocol uses GC to compare the inputs and AHE to compute the rank of each $x_i$ from the comparison bits. We denote an AHE ciphertext with $\lsem \cdot \rsem$ (see Table \ref{Notation_Table}).

\subsection{\skeygc ~Initialization}
\label{Initialization_GC}
The initialization consists of public key distribution and Diffie-Hellman (DH) key agreement. 
Each client $C_i$ sends its public key $pk_i$ (e.g., using a pseudonym certificate) of an AHE to the server. The server then distributes the public keys to the clients. 
In our implementation, we use the Paillier \cite{Paillier.1999} scheme, but any AHE scheme such as \cite{Koblitz1987} will work as well.
Then each pair $(C_i, C_j)$ of clients runs DH key exchange through the server to generate a common secret key $ck_{ij} = ck_{ji}$. 
The common key $ck_{ij}$ is used by $C_i$ and $C_j$ to seed the pseudorandom number generator (PRNG) of the garbling scheme that is used to generate a comparison GC for $x_i$ and $x_j$, i.e. $Gb(1^{\secprm}, ck_{ij}, f_{>})$, where $f_{>}$ is a Boolean comparison circuit.

\subsection{\skeygc ~Main Protocol}
\label{Main_Protocol_GC}
Protocol \ref{YGC_KRE_Protocol} is a four-round protocol in which we use GC to compare pairs of inputs and to reveal a blinded comparison bit to the server. Then we use AHE to unblind the comparison bits, compute the ranks and the KRE without revealing anything to the parties. Let $f_{>}$ be defined as: $f_{>}((\bb_i, x_i), (\bb_j, x_j)) = \bb_i \oplus \bb_j \oplus b_{ij}$, where $\bb_i, \bb_j \in \{0,1\}$, i.e., $f_{>}$ compares $x_i, x_j$ and blinds the comparison bits $b_{ij}$ with $\bb_i, \bb_j$. 

\paragraph{\textbf{Comparing Inputs.}}
For each pair $(x_i, x_j)$, if $\paired(i,j) = true$ the parties do the following:
\begin{itemize}
	\item Client $C_i$ chooses a masking bit $\bb^{ij}_i \stackrel{\$}{\leftarrow} \{0,1\}$ and extends its input to $(\bb^{ij}_i, x_i)$. 
	Then using the common key $ck_{ij}$, it computes $(F^{ij}_{>}, e) \leftarrow Gb(1^{\secprm}, ck_{ij}, f_{>})$ and $(\bar{\bb}_{i}^{ij}, \bar{x}_{i}^{ij}) \leftarrow En(e, (\bb^{ij}_i, x_i))$, and sends $F^{ij}_{>}$, $(\bar{\bb}_{i}^{ij},\bar{x}_{i}^{ij})$ to the server $S$.
	
	\item Client $C_j$ chooses a masking bit $\bb^{ij}_j \stackrel{\$}{\leftarrow} \{0,1\}$ and extends its input $x_j$ to $(\bb^{ij}_j, x_j)$. 
	Then using the common key $ck_{ji} = ck_{ij}$, it computes $(F^{ij}_{>}, e) \leftarrow Gb(1^{\secprm}, ck_{ji}, f_{>})$ and $(\bar{\bb}_{j}^{ij}, \bar{x}_{j}^{ij}) \leftarrow En(e, (\bb^{ij}_j, x_j))$, and sends only $(\bar{\bb}_{j}^{ij},\bar{x}_{j}^{ij})$ to the server $S$.

	\item We have $b'_{ij} \leftarrow F^{ij}_{>}((\bar{\bb}_{i}^{ij}, \bar{x}_{i}^{ij}), (\bar{\bb}_{j}^{ij}, \bar{x}_{j}^{ij})) = \bb^{ij}_i \oplus \bb^{ji}_j \oplus b_{ij} $ (i.e. $b_{ij}$ is hidden to $S$). The server then evaluates all GCs (Steps \ref{ygc_GC_evaluation} to \ref{ygc_GC_evaluation1}).
\end{itemize}

\paragraph{\textbf{Unblinding Comparison Bits.} }
Using AHE, the parties unblind each $b'_{ij} = \bb^{ij}_i \oplus \bb^{ji}_j \oplus b_{ij}$, where $\bb^{ij}_i$ is known to $C_i$ and $\bb^{ij}_j$ is known to $C_j$, without learning anything. As a result $\lsem b_{ij} \rsem_i$ and $\lsem b_{ji} \rsem_j$ are revealed to $S$ encrypted under $pk_i$ and $pk_j$. This is illustrated in Steps \ref{ygc_Comparison_bit} to \ref{ygc_Comparison_bit1} and works as follows:
\begin{itemize}
	\item $S$ sends $b'_{ij}$ to $C_i$ and $C_j$. They reply with $\lsem \bb^{ji}_j \oplus b_{ij} \rsem_i$ and $\lsem \bb^{ij}_i \oplus b_{ij} \rsem_j$.
	\item $S$ forwards $\lsem \bb^{ij}_i \oplus b_{ij} \rsem_j$, $\lsem \bb^{ji}_j \oplus b_{ij} \rsem_i$ to $C_i$, $C_j$. They reply with $\lsem b_{ij} \rsem_j$, $\lsem b_{ij} \rsem_i$.
	\item $S$ sets $\lsem b_{ji} \rsem_j = \lsem 1-b_{ij} \rsem_j$. 
\end{itemize}

\paragraph{\textbf{Computing the Rank.}}
The computation of the rank is done at the server by homomorphically adding comparison bits. Hence for each $i$, the server computes $\lsem r_{i} \rsem_i = \lsem \sum_{j=1}^{n} b_{ij} \rsem_i$. Then, it chooses a random number $\alpha_i$ and computes $\lsem \beta_i\rsem_i = \lsem (r_{i}-k) \cdot \alpha_i \rsem_i$ (Steps \ref{ygc_compute_rank} to \ref{ygc_compute_rank1}). The ciphertext $\lsem \beta_i\rsem_i$ encrypts 0 if $r_i=k$ (i.e., $x_i$ is the \krsym~element) otherwise it encrypts a random plaintext. 

\paragraph{\textbf{Computing the $\kre$'s Ciphertext.}}
Each client $C_i$ receives $\lsem \beta_i\rsem_i$ encrypted under its public key $pk_i$ and decrypts it. Then if $\beta_i = 0$, $C_i$ sets $m_i = x_i$ otherwise $m_i = 0$. Finally, $C_i$ encrypts $m_i$ under each client's public key and sends $\lsem m_i \rsem_1, \ldots, \lsem m_i \rsem_n$ to the server (Steps \ref{ygc_compute_kre} to \ref{ygc_compute_kre1}). 

\paragraph{\textbf{Revealing the $\kre$'s Ciphertext.}}
In the final steps (Steps \ref{ygc_revealing_KRE} to \ref{ygc_revealing_KRE1}), the server adds all $\lsem m_j \rsem_i$ encrypted under $pk_i$ and reveals $\lsem \sum_{j=1}^{n} m_{j} \rsem_i$ to $C_i$.


\skeygc~protocol correctly computes the KRE. The proof trivially follows from the correctness of the GC protocol, Lemmas \ref{rank_computation_lemma} and \ref{Paired_Predicate_Lemma} and the correctness of the AHE scheme.
%
\skeygc~is not fault-tolerant and a collusion with the server reveals all inputs to the adversary.

\begin{figure}[tbp]
  \renewcommand{\figurename}{Protocol}
	\hrule
		\caption{\skeygc~ Protocol}
		\vspace{-2.99mm}
		\hrule
		\begin{algorithmic}[1]
			\For {$i := 1,j := i+1$ \textbf{to}  $n$} \label{ygc_GC_evaluation}
					\If {\paired($ i, j$)} 
						\State $C_i \rightarrow S$: $F^{ij}_{>}, (\bar{\bb}_{i}^{ij}, \bar{x}_i^{ij})$
						\State $C_j \rightarrow S$: $(\bar{\bb}_{j}^{ij}, \bar{x}_j^{ij})$
						\State $S$: \textbf{let} $b'_{ij} \gets F^{ij}_{>}(\bar{x}_i^{ij}, \bar{x}_j^{ij})$ 
					\EndIf
			\EndFor \label{ygc_GC_evaluation1}
				
			\For {$i :=1, j := i+1$ \textbf{to}  $n$} \label{ygc_Comparison_bit}
					\If {\paired($ i, j$)}   
						\State $S \rightarrow C_i$: $b'_{ij}=\bb^{ij}_i \oplus \bb^{ji}_j \oplus b_{ij}$
						\State $S \rightarrow C_j$: $b'_{ij}=\bb^{ij}_i \oplus \bb^{ji}_j \oplus b_{ij}$
						
						\State $C_i \rightarrow S$: $\lsem \bb^{ji}_j \oplus b_{ij} \rsem_i$
						\State $C_j \rightarrow S$: $\lsem \bb^{ij}_i \oplus b_{ij} \rsem_j$
						
						\State $S \rightarrow C_i$: $\lsem \bb^{ji}_i \oplus b_{ij} \rsem_j$
						\State $S \rightarrow C_j$: $\lsem \bb^{ij}_j \oplus b_{ij} \rsem_i$
						
						\State $C_i \rightarrow S$: $\lsem b_{ij} \rsem_j$
						\State $C_j \rightarrow S$: $\lsem b_{ij} \rsem_i$
						
						\State $S$: \textbf{let} $\lsem b_{ji} \rsem_j \gets \lsem 1 - b_{ij} \rsem_j$ 
					\EndIf
			\EndFor \label{ygc_Comparison_bit1}

			\For {$i := 1$ \textbf{to}  $n$} \label{ygc_compute_rank}
					\State $S:$ $\lsem r_{i} \rsem_i \leftarrow \lsem \sum_{j=1}^{n} b_{ij} \rsem_i$ \Comment{$b_{ii} = 1$}
					\State $S \rightarrow C_i$: $\lsem \beta_i \rsem_i \leftarrow \lsem (r_{i} - k)\cdot \alpha_i \rsem_i$, for a random $\alpha_i$
			\EndFor \label{ygc_compute_rank1}
			
			\For {$i := 1$ \textbf{to}  $n$} \label{ygc_compute_kre}
				\State $C_i$: $ m_i :=\begin{cases}
																						 x_i   & \mbox{if} ~ \beta_i = 0\\  
																						 0     & \mbox{if} ~ \beta_i \neq 0
																						\end{cases} $
				\State $C_i \rightarrow S$: $\lsem m_i \rsem_1, \ldots, \lsem m_i \rsem_n$
			\EndFor \label{ygc_compute_kre1}
			
			\For {$i := 1$ \textbf{to}  $n$} \label{ygc_revealing_KRE}
					\State $S \rightarrow C_i$: $\lsem \sum_{j=1}^{n} m_{j} \rsem_i$
			\EndFor \label{ygc_revealing_KRE1}
		\end{algorithmic}
		\label{YGC_KRE_Protocol}
  \hrule
\end{figure}

\section{Protocol \skeahel}
\label{AHE1_Secure_KRE}
This section describes  \skeahel~ (Protocol \ref{AHE1_KRE_Protocol}) based on threshold AHE. \skeahel~ compares all inputs (using our modified variant of the Lin-Tzeng protocol \cite{LinT05}) at the server which then randomly distributes encrypted comparison bits to the  clients for threshold decryption.

\subsection{\skeahel ~Initialization}
\label{Initialization_AHE1}
We assume threshold key generation. Hence, there is a public/private key pair $(pk, sk)$ for an AHE, where  the private key $sk$ is split in $n$ shares $\skshare_1, \ldots, \skshare_n$ such that client $C_i$ gets share $\skshare_i$ and at least $t$ shares are required to reconstruct $sk$. 
Additionally, each client $C_i$ has its own AHE key pair $(pk_i, sk_i)$ and publishes $pk_i$ to all clients. We denote by $\lsem x_i \rsem, \lsem x_i \rsem_j$ encryptions of $x_i$ under $pk, pk_j$ respectively (Table \ref{Notation_Table}). 

\subsection{Modified Lin-Tzeng Comparison Protocol}
\label{LT_Homomorphic_Comparison}

We first describe our modified version of the Lin-Tzeng comparison protocol~\cite{LinT05}.
The main idea of their construction is to reduce the greater-than comparison to the set intersection problem of prefixes.
\paragraph{\textbf{Input Encoding.}}
Let $\textsc{Int}(y_{\eta}\cdots y_{1}) = y$ be a function that takes a bit string of length $\eta$ and parses it into the $\eta-$bit integer $y = \sum_{l=1}^{\eta}y_{l} \cdot 2^{l-1}$. 
The \emph{0-encoding} $V_{x_i}^{0}$ and \emph{1-encoding} $V_{x_i}^{1}$ of an integer input $x_i$ are the following vectors:
$V_{x_i}^{0} = (v_{i\inputlen}, \cdots, v_{i1}), V_{x_i}^{1} = (u_{i\inputlen}, \cdots, u_{i1})$, such that for each $l, (1\leq l \leq \inputlen)$  
\[
v_{il}=
		\begin{cases}
		\textsc{Int}(x_{i\inputlen}x_{i\inputlen-1}\cdots x_{il'}1) & \mbox{if} ~ x_{il} = 0\\  
		r^{(0)}_{il}                         & \mbox{if} ~ x_{il} = 1
		\end{cases} 
\]

\[		
u_{il}=
		\begin{cases}
		\textsc{Int}(x_{i\inputlen}x_{i\inputlen-1}\cdots x_{il}) & \mbox{if} ~ x_{il} = 1\\  
		r^{(1)}_{il}                      & \mbox{if} ~ x_{il} = 0,
		\end{cases} 
\]

\noindent where $l'=l+1$, $r^{(0)}_{il}$, $r^{(1)}_{il}$ are random numbers of a fixed bitlength $\nu > \inputlen$ (e.g. $2^{\inputlen} \leq r^{(0)}_{il}, r^{(1)}_{il} < 2^{\inputlen+1}$) with $LSB(r^{(0)}_{il}) = 0$ and  $LSB(r^{(1)}_{il}) = 1$ (LSB is the least significant bit). If the \textsc{Int} function is used the compute the element at position $l$, then we call it a \emph{proper encoded element} otherwise we call it a \emph{random encoded element}. Note that a random encoded element  $r^{(1)}_{il}$ at position $l$ in the 1-encoding of $x_i$ is chosen such that it is guaranteed to be different to a proper or random encoded element at position $l$ in the 0-encoding of $x_j$, and vice versa. Hence, it is enough if $r^{(1)}_{il}$ and $r^{(0)}_{il}$ are just one or two bits longer than any possible proper encoded element at position $l$. Also note that the bitstring $x_{i\inputlen}x_{i\inputlen-1}\cdots x_{il}$ is interpreted by the function $\textsc{Int}$ as the bitstring $y_{\inputlen-l+1}\cdots y_{1}$ with length $\mu-l+1$ where $y_{1} = x_{il}, y_{2} = x_{i(l+1)}, \ldots, y_{\mu-l+1} = x_{i\mu}$.
If we see $V_{x_i}^{0}, V_{x_j}^{1}$ as sets, then  $x_i > x_j$ iff they have exactly one common element.

\begin{lemma}
 Let $x_i$ and $x_j$ be two integers, then $x_i > x_j ~\mbox{iff}~  V = V_{x_i}^{1} - V_{x_j}^{0}$ has a unique position with 0.
\end{lemma}
\begin{proof}
If $V = V_{x_i}^{1} - V_{x_j}^{0}$ has a unique  0 at a position $l, (1\leq l \leq \inputlen)$ then $u_{il}$ and $v_{il}$ have bit representation $y_{\inputlen-l+1}\cdots y_{1}$, where for each $h, \inputlen-l+1 \geq h \geq 2$, $y_{h} = x_{ig} = x_{jg}$ with $g=l+h-1$, and $y_{1} = x_{il} = 1$ and $x_{jl} = 0$. It follows that $x_i > x_j$.

\noindent If $x_i > x_j$ then there exists a position $l$ such that for each  $h, \inputlen \geq h \geq l+1$, $x_{ih} = x_{jh}$ and $x_{il} = 1$ and $x_{jl} = 0$. This implies $u_{il} = v_{il}$. 

For $h, \inputlen \geq h \geq l+1$, either $u_{ih}$ bit string is a prefix of $x_{i}$ while $v_{jh}$ is random, or $u_{ih}$ is random while  $v_{jh}$ bit string is a prefix of $x_{j}$. From the choice of $r^{(0)}_{ih}$, $r^{(1)}_{ih}$, we have $u_{ih} \neq v_{ih}$. 

For $h, l-1 \geq h \geq 1$ there are three cases: $u_{ih}$ and $v_{ih}$ (as bit string) are both prefixes of $x_{i}$ and $x_{j}$, only one of them is prefix, both are random. For the first case the difference of the bits at position $l$ and for the other cases the choice of $r^{(0)}_{ih}$ imply that $u_{ih} \neq v_{ih}$.
\end{proof}

\paragraph{\textbf{The Protocol.}}
Let $\lsem V_{x_i}^{0} \rsem = \lsem v_{i\inputlen} \rsem, \ldots, \lsem v_{i1} \rsem$ and $\lsem V_{x_i}^{1} \rsem = \lsem u_{i\inputlen} \rsem, \ldots, \lsem u_{i1} \rsem$ denote the componentwise encryption of $V_{x_i}^{0}$ and $V_{x_i}^{1}$. Each client $C_i$ sends $\lsem V_{x_i}^{0} \rsem, \lsem V_{x_i}^{1} \rsem$ to the server. To determine the comparison bit for $x_i > x_j$, the server evaluates the function $\textsc{LinCompare}(\lsem V_{x_i}^{1} \rsem, \lsem V_{x_j}^{0} \rsem)$ (Algorithm \ref{Lin_Compare_Algo}) which returns $\inputlen$ ciphertexts among which exactly one encrypts zero if an only if $x_i > x_j$.

\paragraph{\textbf{Difference to the original protocol.}} In contrast to the original protocol of \cite{LinT05}, we note the following differences:
\begin{itemize}
	\item Additively HE instead of multiplicative: It suits best with our setting and can be implemented using ElGamal on elliptic curve with better performance and smaller ciphertexts. While decrypting requires solving discrete logarithm, this is however not necessary since we are looking for ciphertexts encrypting 0.
	\item \textsc{Int} function: Instead of relying on a collision-free hash function as \cite{LinT05}, we use the \textsc{Int} function which is simpler to implement and more efficient as it produces smaller values.
	\item Choice of random encoded elements $r^{(0)}_{il}$, $r^{(1)}_{il}$: We choose the random encoded elements as explained above and encrypt them, while the original protocol uses ciphertexts chosen randomly in the ciphertext space. 
	\item Encrypting the encodings on both side: In the original protocol, the evaluator has access to $x_j$ in plaintext and does not need to choose random encoded elements. By encoding as explained in our modified version, we can encrypt both encodings and delegate the evaluation to a third party which is not allowed to have access to the inputs in plaintext.
\end{itemize}

\begin{figure}[tbp]
  \renewcommand{\figurename}{Algorithm}
	\hrule
		\caption{Modified Lin-Tzeng Comparison Protocol}
		\vspace{-2.99mm}
		\hrule
		\begin{algorithmic}[1]
		\Function {\textsc{LinCompare}}{$\lsem V_{x_i}^1 \rsem, \lsem V_{x_j}^0 \rsem$}
				\State \textbf{parse} $\lsem V_{x_i}^1 \rsem$ \textbf{as} $\lsem u_{i\inputlen} \rsem, \ldots, \lsem u_{i1} \rsem$
				\State \textbf{parse} $\lsem V_{x_j}^0 \rsem$ \textbf{as} $\lsem v_{i\inputlen} \rsem, \ldots, \lsem v_{i1} \rsem$
				\For {$l := 1$ \textbf{to}  $\inputlen$} 
					\State $c_{l} = \lsem (u_{il} - v_{jl}) \cdot r_l \rsem$, for a random $r_l$
				\EndFor 
				\State \textbf{choose} permutation $\pi \stackrel{\$}{\leftarrow} \mathfrak{S}_{\inputlen}$
				\State  \Return $\pi(c_{\inputlen}, \cdots, c_{1})$
		\EndFunction
		\end{algorithmic}
		\label{Lin_Compare_Algo}
		\hrule
\end{figure}

\subsection{\skeahel ~Main Protocol}
\label{Main_Protocol_AHE1}
Protocol \ref{AHE1_KRE_Protocol} is a four-round protocol in which the clients send their inputs encrypted using AHE under the common public key $pk$ to the server. 
The server homomorphically evaluates comparison circuits on the encrypted inputs using our modified variant of the Lin-Tzeng protocol \cite{LinT05}.
Then the clients jointly decrypt the comparison results and compute the rank of each $x_i$.  

\paragraph{\textbf{Uploading Ciphertexts.}}
Using the common public key $pk$, each client $C_i$ sends $\lsem x_{i} \rsem, \lsem V_{x_i}^{0} \rsem, \lsem V_{x_i}^{1} \rsem$ to the server as illustrated in Step \ref{ciphertext_upload_AHE1} of Protocol \ref{AHE1_KRE_Protocol}. 

\paragraph{\textbf{Comparing Inputs.}}
The server compares the inputs pairwise by computing $g_{ij} \leftarrow \textsc{LinCompare}(\lsem V_{x_i}^{1} \rsem, \lsem V_{x_j}^{0} \rsem)$ for each $1 \leq i,j, \leq n$. Let  $G$ be the $n \times n$ matrix $$[g_{11}, \ldots, g_{1n}, \ldots, g_{n1}, \ldots, g_{nn}].$$ 
The server chooses $n+1$ permutations $\pi, \pi_1, \ldots, \pi_n \stackrel{\$}{\leftarrow} \mathfrak{S}_n$ that hide the indexes of $g_{ij}$ to the clients during threshold decryption:
$\pi$ permutes the rows of $G$ and 
each $\pi_i, 1\leq i \leq n$ permutes the columns of row $i$.
Let $G'_1, \ldots, G'_n$ be the rows of the resulting matrix $G'$ (after application of the permutations to $G$). Using Algorithm \ref{TH_Decryption_Request}, the server computes for each $C_i$ a $t \times n$ matrix $G^{(i)}$ consisting of the rows:
$G'_{i-t+1 \bmod n}, \ldots, G'_{i-1 \bmod n}, G'_i.$
The matrix $G^{(i)}$ and the list of combiners for rows in $G^{(i)}$ are sent to $C_i$ in Step \ref{TH_Decryption_Request_Call_AHE1}. An example is illustrated in Table \ref{Example_Dec_Req}.

Lemma \ref{TH_Dec_Ciphertext_Distribution} shows that the ciphertexts generated from Algorithm \ref{TH_Decryption_Request} allow to correctly decrypt the matrix $G = [g_{11}, \ldots, g_{nn}]$, i.e., each $g_{ij}$ is distributed to exactly $t$ different clients. By applying the lemma to the set of rows of $G$, the first part shows that each client receives exactly a subset of $t$ different rows of $G$. The second part shows that each row of $G$ is distributed to exactly $t$ different clients which allows a correct threshold decryption of each row.

\begin{lemma}
\label{TH_Dec_Ciphertext_Distribution}
Let $\mathbb{X} = \{x_1, \ldots, x_n\}$ be a set of $n$ elements, $\mathbb{X}_i = \{ x_{i-t+1}, \ldots, x_{i} \}$, $1\leq i \leq n$, where the indexes in $\mathbb{X}_i$ are computed modulo $n$, and $t \leq n$. Then:
\begin{itemize}
	\item \label{TH_Dec_Ciphertext_Distribution_1} Each subset $\mathbb{X}_i$ contains exactly $t$ elements of $\mathbb{X}$ and 
	\item \label{TH_Dec_Ciphertext_Distribution_2} Each $x \in \mathbb{X}$ is in exactly $t$ subsets $\mathbb{X}_i$.
\end{itemize}
\end{lemma}
\begin{proof} It is clear from the definition that $\mathbb{X}_i \subseteq \mathbb{X}$ for all $i$ and since $i-(i-t+1) + 1 = t$, $\mathbb{X}_i$ has exactly $t$ elements. 
Let $x_i$ be in $\mathbb{X}$, then from the definition, $x_i$ is element of only the subsets $\mathbb{X}_{i}, \mathbb{X}_{i+1}, \ldots, \mathbb{X}_{i+t-1}$, where indexes of the $\mathbb{X}_i$ are computed $ \bmod ~n$. Again, it holds $(i+t-1)-i + 1 = t$.
\end{proof}

\begin{table}[tbp]
\centering
\begin{tabular}{| c || c || c || c |}
\hline
$G$& $G'$ & $\partdec$ & $\findec$ \\ \hline
$g_{11},  ~g_{12},  ~g_{13}$ & $g_{32}, ~g_{33}, ~g_{31}$ & $~C_1, ~C_2~$ & $C_1$\\
$g_{21},  ~g_{22},  ~g_{23}$ & $g_{11}, ~g_{13}, ~g_{12}$ & $~C_2, ~C_3~$ & $C_2$\\
$g_{31},  ~g_{32},  ~g_{33}$ & $g_{23}, ~g_{21}, ~g_{22}$ & $~C_1, ~C_3~$ & $C_3$\\
\hline
\end{tabular}
\caption{Threshold Decryption Example ($n=3, t=2$): \textnormal{The elements of $G$ are permuted resulting in $G'$. Clients in columns \enquote{\partdec} run \partdec() on the corresponding row and send the result to the client in column \enquote{\findec} for the final decryption.}}
\label{Example_Dec_Req}
\end{table}

After receiving $G^{(i)}$, each client $C_i$ performs its partial decryption for each ciphertext, re-encrypts each line $l ~(l \in I^{(i)})$ with the public key $pk_l$ of client $C_l$. This prevents the server to learn comparison bits. Then $C_i$ sends the result $h_{l, j}^{(i)} = \lsem \partdec(\skshare_i, g_{lj}) \rsem_l, ~(1\leq j \leq n)$ to the server (Step \ref{Client_Th_Dec_AHE1}). Client $C_l$ will be the combiner of the ciphertexts in line $l$. In Step \ref{Server_forward_Th_Dec_AHE1}, the server forwards the encrypted partial decryption results $h_{lj}^{(i_1)}, \ldots, h_{lj}^{(i_t)}$ of line $l$ and the corresponding $c_l = \lsem x_{\pi^{-1}(l)} \rsem$ to $C_l$. Client $C_l$ decrypts and reconstructs each comparison result resulting in the comparison bits $b_{l1}, \ldots, b_{ln}$ as illustrated in Steps \ref{Client_Dec_Comp_Res_AHE1} and \ref{Client_Rec_Comp_Res_AHE1}. 

\begin{figure}[tbp]
  \renewcommand{\figurename}{Algorithm}
	\hrule
		\caption{Decryption Request in \skeahel~and \skeahed}
		\vspace{-2.99mm}
		\hrule
		\begin{algorithmic}[1]
		\Function {\textsc{DecReq}}{$G, i, t, \pi$}
				\State \textbf{parse} $G$ \textbf{as} $[g_{11}, \ldots, g_{1n}, \ldots, g_{n1}, \ldots]$
				\State \textbf{let} $G^{(i)} = [q^{(i)}_{11}, \ldots, q^{(i)}_{1n}, \ldots, q^{(i)}_{t1}, \ldots]$
				\For {$u := 1$ \textbf{to}  $t$} 
					\State $j \gets i-t+u \bmod n$
					\If {$j \leq 0$}
						\State $j \gets j + n$ \Comment {$1 \leq j \leq n$}
					\EndIf
					\State $I^{(i)} \leftarrow I^{(i)} \cup \{j\} $
					\State $w \gets \pi(j)$
					\For {$v := 1$ \textbf{to}  $n$}
						\State $q^{(i)}_{uv} \leftarrow g'_{w v}$ 
					\EndFor 
				\EndFor 
				\State  \Return $(G^{(i)}, I^{(i)})$ 
		\EndFunction
		\end{algorithmic}
		\label{TH_Decryption_Request}
		\hrule
\end{figure}

\paragraph{\textbf{Computing the $\kre$'s Ciphertext.}}
Each combiner $C_l$ computes the rank $r_l = \sum_{j = 1}^n b_{lj}$ (Step \ref{Client_Compute_Rank_AHE1}) and ciphertext $\tilde{c}_l$ that is either a re-encryption of $c_l$ if $r_l = k$ or an encryption of 0 otherwise (Step \ref{Client_Send_KRE_AHE1}). The ciphertext $\tilde{c}_l$ is sent back to the server. The server multiplies all $\tilde{c}_l$ (Step \ref{Server_Compute_KRE_Encrypted_AHE1}) resulting in a ciphertext $\tilde{c}$ of the KRE which is sent to a subset $\mathbb{I}_t$ of $t$ clients for threshold decryption (Step \ref{Server_Send_KRE_Encrypted_AHE1}). Each client in $\mathbb{I}_t$ performs a partial decryption (Step \ref{Client_Part_Dec_KRE_AHE1}), encrypts the result for all clients and sends the ciphertexts to the server (Step \ref{Client_Send_Part_Dec_KRE_AHE1}). Finally, the sever forwards the encrypted partial decryption to the clients (Step \ref{Server_Reveal_KRE_AHE1}) that they use to learn the KRE (Step \ref{Client_Rec_KRE_AHE1}).

\begin{figure}[tbp]
  \renewcommand{\figurename}{Protocol}
	\hrule
    	
		\caption{\skeahel~ Protocol}
		\vspace{-2.99mm}
		\hrule
		\begin{algorithmic}[1]
			\For {$i := 1$ \textbf{to}  $n$}
					\State $C_i \rightarrow S$: $\lsem x_i \rsem, \lsem V_{x_i}^0 \rsem, \lsem V_{x_i}^1 \rsem $ \label{ciphertext_upload_AHE1} 
			\EndFor 
			
			\For {$i,j := 1$ \textbf{to}  $n ~(i \neq j)$}
					\State $S$: $g_{ij} \leftarrow \textsc{LinCompare}(\lsem V_{x_i}^1 \rsem, \lsem V_{x_j}^0 \rsem)$ \label{HOM_Compare_Call_AHE1}
			\EndFor 

			\State $S:$ \textbf{let} permutations $\pi_0, \ldots, \pi_t \stackrel{\$}{\leftarrow} \mathfrak{S}_{n}$
			\State $S:$ \textbf{let} $G = [g_{11}, \ldots, g_{nn}]$ 
			\State $S:$ $(c_{1} \ldots, c_{n}) \leftarrow \pi_0(\lsem x_{1} \rsem, \ldots, \lsem x_{n} \rsem)$
			\For {$u, v := 1$ \textbf{to}  $n$} 
				\State $S:$ $g_{uv} \leftarrow g_{\pi_0(u)\pi_u(v)}$
			\EndFor

			\For {$i := 1$ \textbf{to}  $n$} 
					\State $S \rightarrow C_i$: $Q^{(i)} \leftarrow \textsc{DecReq}(G, i, t, \pi_0$) \label{TH_Decryption_Request_Call_AHE1} %
					\State $C_i$:  \textbf{parse} $Q^{(i)}$ \textbf{as} $(G^{(i)}, I^{(i)})$
					\Statex \hspace{.57cm}  \textbf{parse} $G^{(i)}$ \textbf{as} $[q_{u_11}^{(i)}, \ldots \ldots, q_{u_tn}^{(i)}]$ 
					\Statex \hspace{.57cm}  \textbf{parse} $I^{(i)}$ \textbf{as} $[I^{(i)}_1, \ldots, I^{(i)}_t]$
			\EndFor 

			\For {$i := 1$ \textbf{to}  $n$} 
				\For {$u := 1 $ \textbf{to} $t$} 
					\For {$v := 1$ \textbf{to}  $n$} 
						\ForEach {$l \in I^{(i)}_u$}
							\State $C_i \rightarrow S$: $h_{uv}^{(i)} \leftarrow \lsem \partdec(\skshare_i, q^{(i)}_{uv}) \rsem_l$ \label{Client_Th_Dec_AHE1}
						\EndFor
					\EndFor
				\EndFor
			\EndFor	
			\State \textbf{let} $\{l_1, \ldots, l_t\}$ be the indexes of partial decryptors of the $l$-th row of $G$		
			\For {$l := 1$ \textbf{to}  $n$} \label{Server_forward_Th_Dec_AHE1}
				\State $S \rightarrow C_l$: $c_{l}$ 
				\For {$j := 1$ \textbf{to}  $n$}
					\State $S \rightarrow C_l$: $h^{(l_1)}_{lj}, \ldots, h^{(l_t)}_{lj}$ 
				\EndFor
			\EndFor
			
			\For {$l, j := 1$ \textbf{to}  $n$}
					\State $C_l$: $d_u = Dec(sk_l, h_{lj}^{(l_u)}), u = 1 \ldots t $ \label{Client_Dec_Comp_Res_AHE1}
					\State $C_l$: $b_{lj} \leftarrow \findec(d_1, \ldots, d_t)$ \label{Client_Rec_Comp_Res_AHE1} 
					
			\EndFor
			
			\For {$l := 1$ \textbf{to}  $n$}
					\State $C_l$: $r_l \leftarrow \sum_{j = 1}^n b_{lj}$ \label{Client_Compute_Rank_AHE1}
					\State $C_l \rightarrow S$: $ \tilde{c}_l :=\begin{cases}
																						 c_l \cdot \lsem 0 \rsem  & \mbox{if} ~ r_l = k\\  
																						 \lsem 0 \rsem            & \mbox{if} ~ r_l \neq k
																						\end{cases} 
																						$\label{Client_Send_KRE_AHE1}
			\EndFor
			
			\State $S$: \textbf{let} $\tilde{c} \leftarrow \prod_{l=1}^{n} \tilde{c}_l$ \label{Server_Compute_KRE_Encrypted_AHE1}
			
			\State $S$: \textbf{let} $\mathbb{I}_t = \{i_1, \ldots, i_t\} \stackrel{\$}{\leftarrow} \{1, \ldots, n\}$ \label{Server_Decryptor_Clients_AHE1}

			\ForAll {$i \in \mathbb{I}_t$} 
				\State $S \rightarrow C_i$: $\tilde{c}$  \label{Server_Send_KRE_Encrypted_AHE1}
			\EndFor
			
			\ForAll {$i \in \mathbb{I}_t$}
					\State $C_i$: $m^{(i)} = \partdec(\skshare_i, \tilde{c})$ \label{Client_Part_Dec_KRE_AHE1}
					\State $C_i \rightarrow S$: $\lsem m^{(i)} \rsem_1, \ldots, \lsem m^{(i)} \rsem_n$ \label{Client_Send_Part_Dec_KRE_AHE1}
			\EndFor

			\For {$i := 1$ \textbf{to}  $n$} 
					\State $S \rightarrow C_i$: $\lsem m^{(i_1)} \rsem_i, \ldots, \lsem m^{(i_t)} \rsem_i$ \label{Server_Reveal_KRE_AHE1}
					\State $C_i$: $ \findec(m^{(i_1)}, \ldots, m^{(i_t)})$ \label{Client_Rec_KRE_AHE1}
			\EndFor
			
		\end{algorithmic}
		\label{AHE1_KRE_Protocol}
  \hrule
\end{figure}


\skeahel~protocol correctly computes the KRE. The proof trivially follows from the correctness of the Lin-Tzeng comparison protocol \cite{LinT05}, Lemmas \ref{rank_computation_lemma} and \ref{TH_Dec_Ciphertext_Distribution} and the correctness of AHE.
\skeahel~executes all $O(n^2)$ comparisons non-interactively at the server, but requires threshold decryption for $O(n^2)$ elements. 
The next protocol runs the $O(n^2)$ comparisons in parallel with the help of the clients while requiring threshold decryption of only $O(n)$ elements.

\section{Protocol \skeahed}
\label{AHE2_Secure_KRE}
In this section, we describe \skeahed~ (Protocol \ref{AHE2_KRE_Protocol}) which instantiates the comparison with the DGK protocol \cite{DamgardGK.2007}. The initialization is similar to the previous case.
We start by briefly reviewing the DGK protocol~\cite{DamgardGK.2007}.

\subsection{DGK Comparison Protocol}
Let $(pk_i, sk_i)$ be the public/private key pair of $C_i$. 
Client $C_i$ will be called \textit{Generator} and $C_j$ \textit{Evaluator}. 
Privately evaluating $x_i \geq x_j$ works as follows: 
\begin{itemize}
	\item $C_i$ sends $\lsem x_{i\inputlen} \rsem, \ldots, \lsem x_{i1} \rsem$ to party $C_j$
	\item $C_j$ computes $(\delta_{ji}, \lsem z_{\inputlen} \rsem, \ldots, \lsem z_{1} \rsem) \leftarrow \textsc{\dgkeva}(\lsem x_{i}^{b} \rsem, x_{j}^{b})$ (Algorithm \ref{DGK_Protocol}) and
	sends $(\lsem z_{\inputlen} \rsem, \ldots, \lsem z_{1} \rsem)$ to $C_i$ and outputs $\delta_{ji}$
	\item $C_i$ computes $\delta_{ij} \leftarrow \dgkdec(\lsem z_{\inputlen} \rsem, \ldots, \lsem z_{1} \rsem)$ (Algorithm \ref{DGK_Protocol}) and outputs $\delta_{ij}$. 
\end{itemize}
\begin{figure}[tbp]
  \renewcommand{\figurename}{Algorithm}
	\hrule
		\caption{Algorithms of the DGK Comparison Protocol}
		\vspace{-2.99mm}
		\hrule
		\begin{algorithmic}[1]
		\Function {\dgkcompare}{$ i, j$}
				\If {\paired($ i, j$)} 
					\State $S \rightarrow C_j$: $ \lsem x_{i}^{b} \rsem_i = \lsem x_{i\inputlen} \rsem_i, \ldots, \lsem x_{i1} \rsem_i$
					\State $C_j$: $(\delta_{ji}, Z) \leftarrow \dgkeva(\lsem x_{i}^{b} \rsem_i, x_{j}^{b})$
					
					\State $C_j$: \textbf{parse} $Z$ \textbf{as} $(\lsem z_{i\inputlen} \rsem_i, \ldots, \lsem z_{i1} \rsem_i)$
					
					\State $C_j \rightarrow S$: $ \lsem z_{i\inputlen} \rsem_i, \ldots, \lsem z_{i1} \rsem_i, \lsem \delta_{ji} \rsem$
					\State $S \rightarrow C_i$: $ \lsem z_{i\inputlen} \rsem_i, \ldots, \lsem z_{i1} \rsem_i, \lsem \delta_{ji} \rsem$
					\State $C_i$: $ \delta_{ij} \leftarrow \dgkdec(\lsem z_{i\inputlen} \rsem_i, \ldots, \lsem z_{i1} \rsem_i)$
					\State $C_i \rightarrow S$: $ \lsem \delta_{ij} \oplus \delta_{ji} \rsem$
					\State $S:$ \Return $\lsem b_{ij} \rsem = \lsem \delta_{ij} \oplus \delta_{ji} \rsem$ 
				\EndIf
		\EndFunction
		\end{algorithmic}
		\begin{algorithmic}[1]
		\Function {\dgkeva}{$\lsem x_{i}^{b} \rsem, x_{j}^{b}$}
				\State \textbf{parse} $\lsem x_{i}^{b} \rsem$ \textbf{as} $\lsem x_{i\inputlen} \rsem, \ldots, \lsem x_{i1} \rsem$
				\State \textbf{parse} $x_{j}^{b}$ \textbf{as} $x_{j\inputlen} \ldots x_{j1}$
				\For {$u := 1$ \textbf{to}  $\inputlen$} 
					\State \textbf{compute} $\lsem x_{iu} \oplus x_{ju} \rsem$
				\EndFor 
				\State $\delta_{ji} \stackrel{\$}{\leftarrow} \{0,1\}$, $s \gets 1-2\cdot\delta_{ji}$
				\For {$u := 1$ \textbf{to}  $\inputlen$} 
					\State $\lsem z_u \rsem \leftarrow  (\prod_{v=u+1}^{\inputlen} \lsem x_{iv} \oplus x_{jv} \rsem)^{3}$
					\State $\lsem z_u \rsem \leftarrow \lsem z_u \rsem \lsem s \rsem  \lsem x_{iu} \rsem  \lsem x_{ju} \rsem^{-1}$
					\State $\lsem z_u \rsem \leftarrow \lsem z_u \rsem^{r_u}$, for a random $r_u$
				\EndFor
				\State \textbf{choose} permutation $\pi \stackrel{\$}{\leftarrow} \mathfrak{S}_{\inputlen}$
				\State  \Return ($\delta_{ji}$, $ \pi(\lsem z_{\inputlen} \rsem, \ldots, \lsem z_{1} \rsem))$
		\EndFunction
		\end{algorithmic}
		\begin{algorithmic}[1]
		\Function {\dgkdec}{$\lsem z_{\inputlen} \rsem, \ldots, \lsem z_{1} \rsem$}
				\For {$u := 1$ \textbf{to}  $\inputlen$} 
					\If {$Dec(sk_i, \lsem z_u \rsem) = 0$} 
						\State \Return $1$
					\EndIf
				\EndFor
				\State  \Return $0$
		\EndFunction
		\end{algorithmic}
		\label{DGK_Protocol}
		\hrule
\end{figure}
\noindent 
The comparison with the server is illustrated in Algorithm \ref{DGK_Protocol}. For each pair $C_i, C_j$ such that $\paired(i,j)$ holds, the clients $C_i$ and $C_j$ run the DGK protocol with the server. The server forwards $\lsem x_{i}^{b} \rsem_i = \lsem x_{i\inputlen} \rsem_i, \ldots, \lsem x_{i1} \rsem_i$ (encrypted under $pk_i$) to $C_j$. Client $C_j$ runs $\dgkeva(\lsem x_{i}^{b} \rsem_i, x_{j}^{b})$ and obtains $\delta_{ji}, (\lsem z_{i\inputlen} \rsem_i, \ldots, \lsem z_{i1} \rsem_i)$ as result. It then encrypts $\delta_{ji}$ under the common public key and sends back $ \lsem z_{i\inputlen} \rsem_i, \ldots, \lsem z_{i1} \rsem_i, \lsem \delta_{ji} \rsem$ to client $C_i$ via the server. Client $C_i$ runs the decryption $\dgkdec(\lsem z_{i\inputlen} \rsem_i, \ldots, \lsem z_{i1} \rsem_i)$, obtains a shared bit $\delta_{ij}$ and sends back $ \lsem \delta_{ij} \oplus \delta_{ji} \rsem$ to the server. After the computation the clients $C_i$ and $C_j$ hold random shared bits $\delta_{ij}$ and $\delta_{ji}$ such that  $b_{ij} = [x_i \leq x_j]  = \delta_{ij} \oplus \delta_{ji}$ holds. The server learns the encryption $\lsem b_{ij} \rsem$ of the comparison bit $b_{ij}$. In the DGK protocol clients $C_i$ and $C_j$ perform respectively $O(\inputlen)$ and $O(6\inputlen)$ asymmetric operations. 

\subsection{\skeahed ~Main Protocol}
\label{Main_Protocol_TH-AHE-SM}
\skeahed~is a 4-round protocol in which inputs are compared interactively using the DGK protocol. The resulting comparison bit is encrypted under 
$pk$ and revealed to the server which then computes the ranks of the $x_i$'s and trigger a threshold decryption. 

\paragraph{\textbf{Uploading Ciphertext.}}
Each party $C_i$ sends $\lsem x \rsem$ (encrypted under the common public key $pk$) and $\lsem x_i^b \rsem_i = (\lsem x_{i\inputlen} \rsem_i, \ldots, \lsem x_{i1} \rsem_i)$ (encrypted under its own public key $pk_i$) to the server. This is illustrated in Step \ref{ciphertext_upload_AHE2} of protocol \ref{AHE2_KRE_Protocol}. The server then initializes a matrix $G = [g_{11}, \ldots, g_{nn}]$, where $g_{ii} = \lsem 1 \rsem$ and $g_{ij} (i\neq j)$ will be computed in the DGK protocol as $g_{ij} = \lsem b_{ij}\rsem$ if $\paired(i,j)$ is true,  and an array $X = [ \lsem x_1 \rsem, \ldots,  \lsem x_n \rsem]$ (Step \ref{init_matrix_AHE}). 

\paragraph{\textbf{Comparing Inputs.}}
In this step, pairs of clients run \dgkcompare~with the server as illustrated in Algorithm \ref{DGK_Protocol}. 
If $(i,j)$ satisfies the predicate \paired, then $C_i$ runs the DGK protocol as generator and $C_j$ is the evaluator. After the computation, $C_i$ and $C_j$ get shares $\delta_{ij}$ and $\delta_{ji}$ of the comparison bit which is encrypted under $pk$ as $ \lsem b_{ij} \rsem = \lsem \delta_{ij} \oplus \delta_{ji} \rsem$ and revealed to the server. 

\paragraph{\textbf{Computing the $\kre$'s Ciphertext.}}
After all admissible comparisons have been computed (and the result stored in the matrix $G$), the server uses Algorithm \ref{Compute_Rank} to compute the rank of each input $x_i$ by homomorphically adding the comparison bits involving $x_i$. Let $\lsem r_i \rsem$ be a ciphertext initially encrypting 0 and let $b_{ij} = \delta_{ij} \oplus \delta_{ji}$. For each $j$, if $\paired(i,j)$ is true (i.e., $\lsem b_{ij} \rsem$ has been computed) then we compute $\lsem r_i \rsem \leftarrow \lsem r_i +  b_{ij}\rsem$. Otherwise (i.e., $\lsem b_{ij} \rsem$ has not been computed but we can deduce it from $\lsem b_{ji} \rsem$) we compute $\lsem r_i \rsem \leftarrow \lsem r_i + 1 - b_{ij}\rsem$.
Now, the server has the encrypted rank $\lsem r_1 \rsem, \ldots, \lsem r_n \rsem$, where exactly one $\lsem r_i \rsem$ encrypts $k$. Since we are looking for the element whose rank is $k$, the server then computes $y_i = (\lsem r_i \rsem \cdot \lsem k \rsem^{-1} )^{\alpha_i} \cdot \lsem x_i \rsem = \lsem (r_i - k) \alpha_i + x_i \rsem$ for all $i$, where $\alpha_i$ is a number chosen randomly in the plaintext space. 
Therefore, for the ciphertext $\lsem r_i \rsem$ encrypting $k$, $y_i$ is equal to $\lsem x_i \rsem$. Otherwise $y_i$ encrypts a random plaintext. 
\begin{figure}[tbp]
  \renewcommand{\figurename}{Algorithm}
	\hrule
		\caption{Computing the KRE's ciphertext in \skeahed}
		\vspace{-2.99mm}
		\hrule
		\begin{algorithmic}[1]
		\Function {\textsc{ComputeKreAhe}}{$G, X, k$}
				\State \textbf{parse} $G$ \textbf{as} $[g_{11}, \ldots, g_{nn}]$ 
				\State \textbf{parse} $X$ \textbf{as} $[ \lsem x_1 \rsem, \ldots,  \lsem x_n \rsem]$
				\For {$i := 1$ \textbf{to}  $n$}
					\State $\lsem r_i \rsem \leftarrow g_{ii}$
					\For {$j := 1$ \textbf{to}  $n~ (j \neq i)$} 
						\If {\paired($ i, j$)} 
							\State $\lsem r_i \rsem \leftarrow \lsem r_i \rsem \cdot g_{ij}$
						\Else
							\State $\lsem r_i \rsem \leftarrow \lsem r_i \rsem \cdot \lsem 1 \rsem \cdot  g_{ji}^{-1}$
						\EndIf
					\EndFor
				\EndFor
				\For {$i := 1$ \textbf{to}  $n$}
					\State $y_i \leftarrow ( \lsem r_i \rsem \cdot \lsem k \rsem^{-1} )^{\alpha_i} \cdot \lsem x_i \rsem$
				\EndFor
				\State  \Return $[y_1, \ldots, y_n]$ 
		\EndFunction
		\end{algorithmic}
		\label{Compute_Rank}
  \hrule
\end{figure}

\paragraph{\textbf{Decrypting the $\kre$'s Ciphertext.}}
In Step \ref{Decryption_Request_Call_AHE2}, the server distributes the result $Y = [y_1, \ldots, y_n]$ of Algorithm \ref{Compute_Rank} to the clients for threshold decryption. For that, the array $Y$ is passed as $n \times 1$ matrix to Algorithm \ref{TH_Decryption_Request}.
In Step \ref{Client_Th_Dec_AHE2}, the server receives partial decryption results from the clients, forwards them to the corresponding combiner (Step \ref{Forward_to_combiner}). Each combiner $C_j$ performs a final decryption (Step \ref{Combiner_Rec_KRE_AHE2}) resulting in a message $\tilde{x}_j$ whose bitlength is less or equal to $\inputlen$ if it is the KRE. Combiner $C_j$ then sets $m^{(j)} = \tilde{x}_j$ if $|\tilde{x}_j| \leq \inputlen$, otherwise  $m^{(j)} = 0$ (Step \ref{combiner_Select_KRE}). Then $m^{(j)}$ is encrypted with the public key of all clients and send to the server (Step \ref{combiner_Send_KRE_Ciphertext_AHE2}). Finally, the server reveals the KRE to all clients (Step \ref{Server_Reveal_KRE_AHE2}).
 
\begin{figure}[!htb]
  \renewcommand{\figurename}{Protocol}
		\hrule
		\caption{\skeahed~ Protocol}
		\vspace{-2.99mm}
		\hrule
		\begin{algorithmic}[1]
			\For {$i := 1$ \textbf{to}  $n$} 
					\State $C_i \rightarrow S$: $ \lsem x_i \rsem, \lsem x_i^b \rsem_i$ \label{ciphertext_upload_AHE2}
			\EndFor

			\State $S:$ \textbf{let} $G = [g_{11}, \ldots, g_{nn}]$  \label{init_matrix_AHE}
			\Statex $S:$ \textbf{let} $X = [ \lsem x_1 \rsem, \ldots,  \lsem x_n \rsem]$ 
			
			\For {$i := 1, j := i+1$ \textbf{to}  $n$} 
					\State $C_i, C_j, S$: $g_{ij} \leftarrow \dgkcompare(i, j)$ \label{DGK_Compare_Call}
			\EndFor
			
			\State $S$: $Y \leftarrow \textsc{ComputeKreAhe}(G, X, k)$ \label{Call_Compute_KRE}
			\State $S$: \textbf{let} $\pi \stackrel{\$}{\leftarrow} \mathfrak{S}_{n}$ be a permutation

			\State $S$: \textbf{parse} $Y$ \textbf{as} $[y_{1}, \ldots, y_{n}]$
			\For {$i := 1$ \textbf{to}  $n$} 
				\State $S$: $y_i \leftarrow y_{\pi(i)}$
			\EndFor

			\For {$i := 1$ \textbf{to}  $n$} 
					\State $S \rightarrow C_i$: $Q^{(i)} \leftarrow \textsc{DecReq}(Y, i, t, \pi )$ \label{Decryption_Request_Call_AHE2}
					\State $C_i$: \textbf{parse} $Q^{(i)}$ as $(Z^{(i)}, I^{(i)})$
					\Statex \hspace{.57cm} \textbf{parse} $Z^{(i)}$ as $[z_{j_1}^{(i)}, \ldots, z_{j_t}^{(i)}]$ 
			\EndFor
			
				\For {$i := 1$ \textbf{to}  $n$} 
					\ForEach {$j$ \textbf{in}  $I^{(i)}$} 
						\State $C_i \rightarrow S$: $h_{j}^{(i)} \leftarrow \lsem \partdec(\skshare_i, z^{(i)}_j) \rsem_j$ \label{Client_Th_Dec_AHE2}
				\EndFor
			\EndFor
			
			\For {$j := 1$ \textbf{to}  $n$} 
					\State $S \rightarrow C_j$: $(h_{j}^{(i_1)}, \ldots, h_{j}^{(i_t)})$ \label{Forward_to_combiner}
			\EndFor

			\For {$j := 1$ \textbf{to}  $n$} 
					\State $C_j$: $d_u = Dec(sk_j, h_{j}^{(i_u)}), u = 1, \ldots, t$ \label{Combiner_Dec_KRE_Shares_AHE2}
					\State $C_j$: $\tilde{x}_j \leftarrow \findec(d_1, \ldots, d_t)$ \label{Combiner_Rec_KRE_AHE2}
					\State $C_j$: $ m^{(j)} :=\begin{cases}
																						 \tilde{x}_j & \mbox{if} ~ |\tilde{x}_{j'}| \leq \inputlen\\  
																						 0           & \mbox{if} ~ |\tilde{x}_{j'}| > \inputlen
																						\end{cases} 
																						$\label{combiner_Select_KRE}
					\State $C_j \rightarrow S$: $\lsem m^{(j)} \rsem_1, \ldots, \lsem m^{(j)} \rsem_n$ \label{combiner_Send_KRE_Ciphertext_AHE2}
			\EndFor

			\For {$i := 1$ \textbf{to}  $n$} 
					\State $S \rightarrow C_i$: $\lsem \sum_{j=1}^{n}m^{(j)} \rsem_i$ \label{Server_Reveal_KRE_AHE2}
					\State $C_i$: $ Dec(sk_i, \lsem \sum_{j=1}^{n}m^{(j)} \rsem_i)$ \label{Client_Dec_KRE_AHE2}
			\EndFor
			
		\end{algorithmic}
		\label{AHE2_KRE_Protocol}
	\hrule
\end{figure}


\skeahed~protocol correctly computes the KRE. This trivially follows from the correctness of DGK protocol \cite{DamgardGK.2007}, Lemmas \ref{rank_computation_lemma} and \ref{TH_Dec_Ciphertext_Distribution} and the correctness of AHE. 
\skeahed~evaluates comparisons interactively but requires threshold decryption for $O(n)$ elements. 
Notice that \skeahed~can be instantiated with the Lin-Tzeng protocol \cite{LinT05} as well. To compare $x_i, x_j$, $C_j$ will receive both $\lsem V^{0}_{x_i} \rsem, \lsem V^{1}_{x_i} \rsem$ and randomly choose between evaluating $\textsc{LinCompare}$ either with $\lsem V^{1}_{x_i} \rsem, \lsem V^{0}_{x_j} \rsem$ or with $\lsem V^{1}_{x_j} \rsem, \lsem V^{0}_{x_i} \rsem$.
It will then set $\delta_{ji}\leftarrow 0$ or $\delta_{ji}\leftarrow 1$ accordingly. This improves the running time 
(\textsc{LinCompare} is more efficient than \textsc{\dgkeva}) 
while increasing the communication ($\inputlen$ more ciphertexts are sent to $C_j$).

In \skeahel~and \skeahed, we evaluated either the comparison (\skeahel) or the rank (\skeahed) completely at the server.
In the next scheme, we compute the KRE's ciphertext non-interactively at the server such that clients are only required for the threshold decryption of one ciphertext.

\section{Protocol \skeshe}
\label{SHE_Secure_KRE}

This section describes \skeshe~based on SHE. Hence, $\lsem x \rsem$ now represents an SHE ciphertext of the plaintext $x$. 
The initialization and threshold decryption are similar to \skeahel. 

\subsection{SHE Routines}

Protocol \ref{FHE_KRE_Protocol} is based on the BGV scheme \cite{BrakerskiGV11} as implemented in HElib \cite{HaleviS14} and requires binary plaintext space and Smart-Vercauteren  ciphertext packing (SVCP) technique \cite{SmartV2014}. 

Using SVCP, a ciphertext consists of a fixed number $m$ of slots encrypting bits, i.e. $\lsem \cdot | \cdot | \ldots | \cdot \rsem$. The encryption of a bit $b$ replicates $b$ to all slots, i.e., $\lsem b \rsem = \lsem b | b | \ldots | b \rsem$. However, we can pack the bits of $x^{b}_{i}$ in one ciphertext and will denote it by $\lsem \vec{x}_i \rsem = \lsem x_{i\inputlen} | \ldots | x_{i1}|0|\ldots|0 \rsem$. 

Each $C_i$ sends $\lsem x^{b}_{i} \rsem, \lsem \vec{x}_i \rsem$ to $S$ as input to Algorithm \ref{FHE_KRE} which uses built-in routines to compute the KRE. 

\noindent The routine \textsc{SheAdd} takes two or more ciphertexts and performs a component-wise addition mod 2, i.e., we have: 
$$\textsc{SheAdd}( \lsem b_{i1} | \ldots | b_{im} \rsem,  \lsem b_{j1} | \ldots | b_{jm} \rsem) =  \lsem b_{i1} \oplus b_{j1} | \ldots | b_{im} \oplus b_{jm} \rsem.$$ 

\noindent Similarly,  for component-wise multiplication routine we have: 
$$\textsc{SheMult}( \lsem b_{i1} | \ldots | b_{im} \rsem,  \lsem b_{j1} | \ldots | b_{jm} \rsem) =  \lsem b_{i1} \cdot b_{j1} | \ldots | b_{im} \cdot b_{jm} \rsem.$$ 

Let $x_i, x_j$ be two integers, $b_{ij} = [x_i > x_j]$ and $b_{ji} = [x_j > x_i]$, the routine \textsc{SheCmp} takes  $\lsem x^{b}_{i} \rsem,  \lsem x^{b}_{j} \rsem$, compares $x_i$ and $x_j$ and returns $\lsem b_{ij} \rsem$, $\lsem b_{ji} \rsem$. Note that if the inputs to \textsc{SheCmp} encrypt the same value, then the routine outputs two ciphertexts of 0.

Let $b_{i1}, \ldots, b_{in}$ be $n$ bits such that $r_i = \sum_{j=1}^{n}b_{ij}$ and let $r_i^b=r_{i\log n}, \ldots, r_{i1}$ be the bit representation of $r_i$. The routine \textsc{SheFadder} implements a full adder on $\lsem b_{i1} \rsem, \ldots, \lsem b_{in} \rsem$ and returns $\lsem r_i^b \rsem = (\lsem r_{i\log n} \rsem, \ldots, \lsem r_{i1} \rsem)$. 

There is no built-in routine for equality check in HElib. We implemented it using \textsc{SheCmp} and  \textsc{SheAdd}. Let $x_i$ and $x_j$ be two $\inputlen$-bit integers. 
We use \textsc{SheEqual} to denote the equality check routine and implement $\textsc{SheEqual}(\ctxtrep{\bitrep{x_i}},  \ctxtrep{\bitrep{x_j}})$ by computing:
\begin{itemize}
	\item $(\lsem b'_{i} \rsem$, $\lsem b''_{i} \rsem) = \textsc{SheCmp}(\ctxtrep{\bitrep{x_i}},  \ctxtrep{\bitrep{x_j}})$ and
	\item  $\lsem \beta_i\rsem = \textsc{SheAdd}(\lsem b'_{i} \rsem$, $\lsem b''_{i} \rsem, \lsem 1 \rsem)$, which results in $\beta_i = 1$ if $x_i = x_j$ and $\beta_i = 0$ otherwise.
\end{itemize}

\subsection{\skeshe~Main Protocol}

In Protocol \ref{FHE_KRE_Protocol} the server $S$ receives encrypted inputs from clients. For each client's integer $x_i$, the encrypted input consists of:
\begin{itemize}
    \item an encryption  $\lsem \bitrep{x_i} \rsem = (\ctxtrep{x_{i\inputlen}}, \ldots, \ctxtrep{x_{i1}})$ of the bit representation and
    \item an encryption $\lsem \vec{x}_i \rsem = \lsem x_{i\inputlen} | \ldots | x_{i1}|0|\ldots|0 \rsem$ of the packed bit representation .
\end{itemize}
Then the server runs Algorithm \ref{FHE_KRE} which uses $\shecmp$ to pairwise compare the inputs resulting in encrypted comparison bits $\ctxtrep{b_{ij}}$. Then $\shefadder$ is used to compute the rank of each input by adding comparison bits. The result is an encrypted bit representation $\lsem \bitrep{r_{i}} \rsem$ of the ranks. Using the encrypted bit representations $\lsem \bitrep{k} \rsem$, $\lsem \bitrep{r_{i}} \rsem$  of $k$ and each rank, $\sheequ$ checks the equality and returns an encrypted bit $\lsem \beta_{i} \rsem$.
Recall that because of SVCP the encryption of a bit $\beta_i$ is automatically replicated in all slots, i.e., $\lsem \beta_i \rsem = \lsem \beta_i | \beta_i | \ldots | \beta_i \rsem$, such that evaluating $\lsem \vec{y_{i}} \rsem \leftarrow \shemul(\lsem \vec{x}_i \rsem, \lsem \beta_i \rsem)$, $1\leq i \leq n$, and $\sheadd(\lsem \vec{y_{1}} \rsem, \ldots, \lsem \vec{y_{n}} \rsem)$ returns the KRE's ciphertext.
Recall that because of SVCP the encryption of a bit $\beta_i$ is automatically replicated in all slots, i.e., $\lsem \beta_i \rsem = \lsem \beta_i | \beta_i | \ldots | \beta_i \rsem$, such that evaluating $\lsem \vec{y_{i}} \rsem \leftarrow \textsc{SheMult}(\lsem \vec{x}_i \rsem, \lsem \beta_i \rsem)$, $1\leq i \leq n$, and $\textsc{SheAdd}(\lsem \vec{y_{1}} \rsem, \ldots, \lsem \vec{y_{n}} \rsem)$ returns the KRE's ciphertext.

Correctness and security follow trivially from Lemma \ref{rank_computation_lemma}, correctness and security of SHE.
The leakage is $\mathcal{L}_S = \mathcal{L}_i = \{n, t, \ctxtlen, \secprm, \inputlen\}$. 

\begin{figure}[!htb]
  \renewcommand{\figurename}{Algorithm}
	\hrule
    	
		\caption{Computing the KRE's Ciphertext in \skeshe}
		\vspace{-2.99mm}
		\hrule
		\begin{algorithmic}[1]
		\Function {\textsc{ComputeKreShe}}{$X, Z, c$}
				\State \textbf{parse} $X$ \textbf{as} $[\lsem x^{b}_1 \rsem, \ldots, \lsem x^{b}_n \rsem]$ 
				\Statex \hspace{.02cm} \textbf{parse} $Z$ \textbf{as} $[\lsem \vec{x}_1 \rsem, \ldots, \lsem \vec{x}_n \rsem]$
				\Statex \hspace{.02cm} \textbf{parse} $c$ \textbf{as} $ \lsem k^{b} \rsem$
				\For {$i := 1$ \textbf{to}  $n$}
					\State $\lsem b_{ii} \rsem \leftarrow \lsem 1 \rsem$
					\For {$j := i+1$ \textbf{to}  $n$} 
						\State $(\lsem b_{ij} \rsem, \lsem b_{ji} \rsem) \leftarrow \textsc{SheCmp}(\lsem x_i^b \rsem, \lsem x_j^b \rsem)$
					\EndFor
				\EndFor
				\For {$i := 1$ \textbf{to}  $n$} 
						\State $\lsem r_{i}^{b} \rsem \leftarrow \textsc{SheFadder}(\lsem b_{i1} \rsem, \ldots, \lsem b_{in} \rsem)$
				\EndFor
				\For {$i := 1$ \textbf{to}  $n$}
					\State $\lsem \beta_{i} \rsem \leftarrow \textsc{SheEqual}(\lsem r_{i}^{b} \rsem, \lsem k^{b} \rsem)$ \label{compute_f_r}
				\EndFor
				\For {$i := 1$ \textbf{to}  $n$}
					\State $\lsem \vec{y_{i}} \rsem \leftarrow \textsc{SheMult}(\lsem \vec{x}_i \rsem, \lsem \beta_i \rsem)$ \label{inner_product}
				\EndFor
				\State  \Return $\textsc{SheAdd}(\lsem \vec{y_{1}} \rsem, \ldots, \lsem \vec{y_{n}} \rsem)$ 
		\EndFunction
		\end{algorithmic}
		\label{FHE_KRE}
  \hrule
\end{figure}

\begin{figure}[!htb]
  \renewcommand{\figurename}{Protocol}
	\hrule
		\caption{\skeshe~ Protocol}
		\vspace{-2.99mm}
		\hrule
		\begin{algorithmic}[1]
					%
			\For {$i := 1$ \textbf{to}  $n$}
					\State $C_i \rightarrow S$: $\lsem x^{b}_i \rsem, \lsem \vec{x}_i \rsem$ \label{ciphertext_upload_TH_FHE} 
			\EndFor 
			
			\State $S:$ \textbf{let} $X = [\lsem x^{b}_1\rsem, \ldots, \lsem x^{b}_n \rsem]$ 
			\Statex \hspace{.39cm} \textbf{let} $Z = [\lsem \vec{x}_1 \rsem, \ldots, \lsem \vec{x}_n \rsem]$
			\Statex \hspace{.39cm} \textbf{let} $c = \lsem k^{b} \rsem$

			\State $S:$ $ c' \leftarrow \textsc{ComputeKreShe}(X, Z, c)$ \label{Server_Compute_KRE_TH_FHE}
			\Statex \hspace{.39cm} \textbf{parse} $c'$ \textbf{as} $ \lsem x_{i^*} \rsem $


			\State $S$: \textbf{let} $\mathbb{I}_t = \{i_1, \ldots, i_t\} \stackrel{\$}{\leftarrow} \{1, \ldots, n\}$ \label{Server_Decryptor_Clients_TH_FHE}

			\ForAll {$i \in \mathbb{I}_t$} 
				\State $S \rightarrow C_i$: $\lsem x_{i^*} \rsem$  \label{Server_Send_KRE_Encrypted_TH_FHE}
			\EndFor
			
			\ForAll {$i \in \mathbb{I}_t$} 
					\State $C_i$: $m^{(i)} \leftarrow \partdec(\skshare_i, \lsem x_{i^*} \rsem)$ \label{Combiner_Rec_KRE_Share_TH_FHE}
					\State $C_i \rightarrow S$: $\lsem m^{(i)} \rsem_1, \ldots, \lsem m^{(i)} \rsem_n$ \label{Client_Send_KRE_Share_TH_FHE}
			\EndFor

			\For {$i := 1$ \textbf{to}  $n$} 
					\State $S \rightarrow C_i$: $\lsem m^{(i_1)} \rsem_i, \ldots, \lsem m^{(i_t)} \rsem_i$ \label{Server_Reveal_KRE_TH_FHE}
					\State $C_i$: $\findec(m^{(i_1)}, \ldots, m^{(i_t)})$ \label{Client_Rec_KRE_TH_FHE}
			\EndFor
		\end{algorithmic}
		\label{FHE_KRE_Protocol}
  \hrule
\end{figure}

\section{Evaluation}
\label{Evaluation}
In this section, we discuss some implementation details and evaluate and compare our schemes.

\subsection{Implementation Details}
\label{Implementation}
We implemented our schemes using software libraries SCAPI \cite{EjgenbergFLL.2012, scapiWeb} and HElib \cite{HaleviS14, HElibWeb} which we briefly describe here.
\paragraph{SCAPI}
We implemented \skeygc, \skeahel, \skeahed~as client-server Java applications while using SCAPI \cite{EjgenbergFLL.2012}. SCAPI is an open-source Java library for implementing SMC protocols. It provides a reliable, efficient, and highly flexible cryptographic infrastructure. It also provides a GC framework including several optimizations such as OT extensions, free-XOR, garbled row reduction \cite{EjgenbergFLL.2012}. This GC framework has been used to run GC comparison circuits in \skeygc. Furthermore, there is a built-in communication layer that provides communication services for any interactive cryptographic protocol. 

\paragraph{HElib}
As \skeshe~ mostly consists of the homomorphic evaluation by the server, we implemented Algorithm \ref{FHE_KRE} using HElib \cite{HaleviS14}. 
HElib is a software library that implements homomorphic encryption. The current version of the library includes an implementation of the BGV lattice-based homomorphic encryption scheme \cite{BrakerskiGV11} . It also includes various optimizations that make homomorphic encryption runs faster, including the Smart-Vercauteren ciphertext packing (SVCP) techniques \cite{SmartV2014}. We also implemented the threshold decryption for a $n$-out-of-$n$ secret sharing of BGV's private key.

\subsection{Experiments}
In this section, we report on the experimental results of our implementations. We start by describing the experimental setup.
\paragraph{Experimental Setup.}
For \skeygc, \skeahel, \skeahed, 
we conducted experiments using for the server a machine with a 6-core Intel(R) Xeon(R) E-2176M CPU @ 2.70GHz and 32GB of RAM, and for the clients two machines with each two Intel(R) Xeon(R) CPU E7-4880 v2 @ 2.50GHz. The client machines were equipped with 8GB and 4 GB of RAM, and were connected to the server via WAN. Windows 10 Enterprise was installed on all three machines.. Windows 10 Enterprise was installed on all three machines.
For each experiment, about 3/5 of the clients were run on the machine with 8 GB RAM while about 2/5 were run on the machine with 4 GB RAM. We ran all experiments using JRE version 8. 

Since the main computation of \skeshe~ is done on the server, we focus on evaluation Algorithm \ref{FHE_KRE} on a Laptop with Intel(R) Core(TM) i5-7300U CPU @ 2.60GHz running 16.04.1-Ubuntu with 4.10.0-14-lowlatency Kernel version. 
\paragraph{Results.}
\begin{figure}[tbp]
    \centering
		\begin{subfigure}[t]{0.450\textwidth}
        \includegraphics[width=55mm]{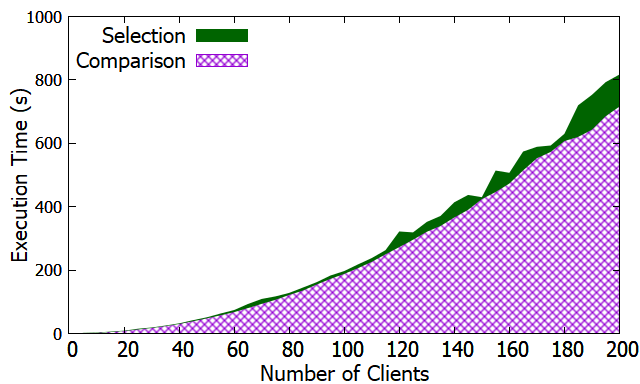}
        \caption{\scriptsize Computation Cost \skeygc}
        \label{Graph_CLT_SKE_YGC_AHE_Computation_Cost}
		\end{subfigure}
    \begin{subfigure}[t]{0.450\textwidth}
        \includegraphics[width=55mm]{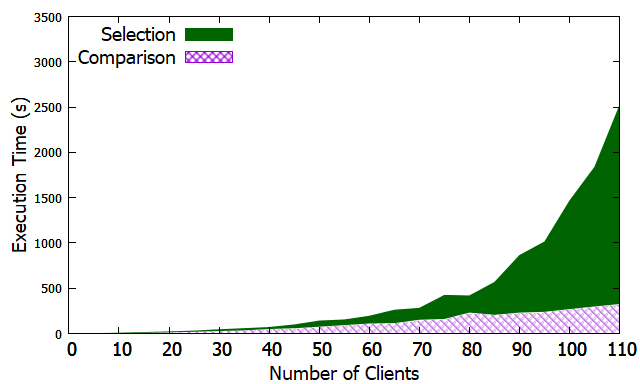}
        \caption{\scriptsize Computation Cost \skeahel}
        \label{Graph_CLT_SKE_AHE1_Computation_Cost}
		\end{subfigure}
    \begin{subfigure}[t]{0.45\textwidth}
        \includegraphics[width=55mm]{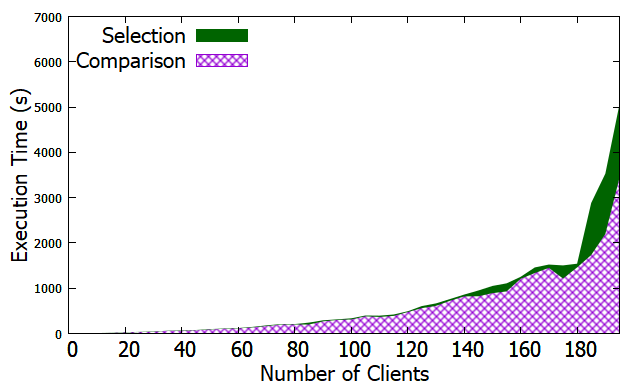}
        \caption{\scriptsize Computation Cost \skeahed}
        \label{Graph_CLT_SKE_AHE2_Computation_Cost}
		\end{subfigure}
		\begin{subfigure}[t]{0.45\textwidth}
        \includegraphics[width=55mm]{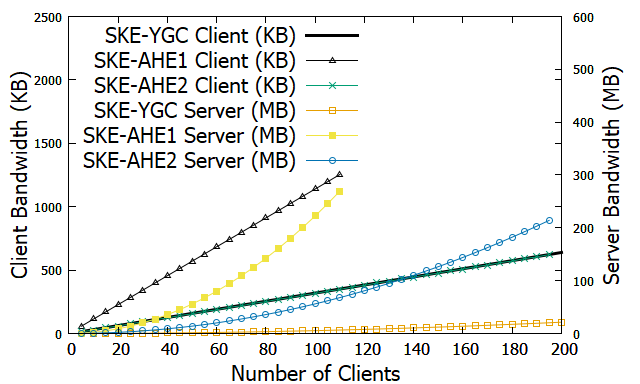}
        \caption{\scriptsize Communication Cost in MB}
        \label{Graph_CLT_SKE_AHE2_Computation_Cost0}
		\end{subfigure}
    \caption{Results for \skeygc, \skeahel, \skeahed}
		\label{Graph_Computation_Costs}
\end{figure}
\begin{figure}[tbp]
    \centering
		\begin{subfigure}[t]{0.450\textwidth}
        \includegraphics[width=55mm]{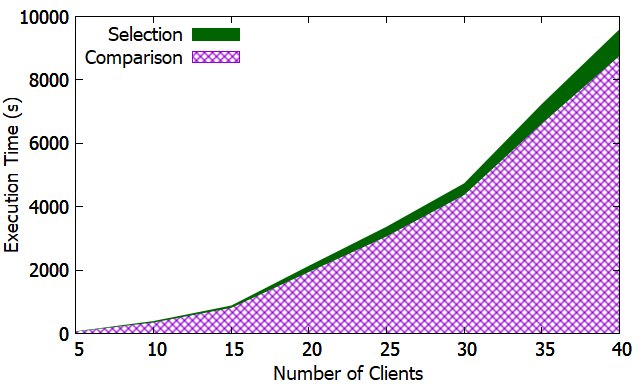}
        \caption{\scriptsize Server Computation Cost}
        \label{Graph_SVR_SKE_SHE_Computation_Cost}
		\end{subfigure}
    \begin{subfigure}[t]{0.450\textwidth}
        \includegraphics[width=55mm]{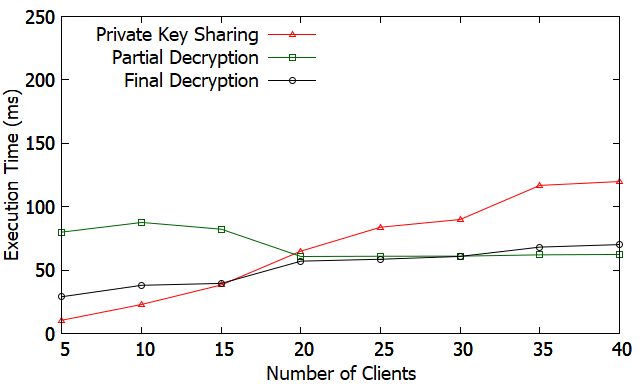}
        \caption{\scriptsize Threshold Decryption Cost}
        \label{Graph_SKE_SHE_THDEC_Cost}
		\end{subfigure}
    \caption{Performance Results \skeshe}
		\label{Graph_SKESHE_Costs}
\end{figure}

We evaluated \skeygc, \skeahel, \skeahed~at security level $\secprm = 128$, bitlength $\inputlen=32$ and (minimal) threshold $t=2$ for threshold decryption.
We instantiated \skeahel~and \skeahed~with Elliptic Curve ElGamal (see Appendix \ref{sec:EcElGamal}) using elliptic curve \texttt{secp256r1}.
Figure 8 shows our performance results which are summarized in Table \ref{Comparision_YGC_DGK_LTT} for $n=100$ clients.
In Table \ref{Comparision_YGC_DGK_LTT}, we also illustrate the costs when $t=1$ (i.e., each $C_i$ knows $sk$) for both \skeahed~and \skeahel~and when $t = n$ (i.e. all $C_i$ must participate in the threshold decryption) for \skeahed.

\skeygc~is the most efficient in both computation and communication and takes 197 seconds to each client to compute the KRE of 100 clients in a WAN setting. 
The communication is 0.31 MB for each client and 5.42 MB for the server. However, \skeygc~is neither collusion-resistant nor fault-tolerant. 

\skeahed~is the second most efficient and is collusion-resistant and fault-tolerant. Although it requires more interactions to compute comparisons, we batched many comparisons together
and were able to run threshold decryption for $O(n)$ elements, instead of $O(n^2)$ as in \skeahel. 
The computation of the KRE of 100 values takes to each client 353 seconds (for $t=1$), 336 seconds (for $t=2$) and 441 seconds (for $t=100$). 
The communication is 0.3 MB, 0.3 MB, 0.32 MB for each client and 56.07 MB, 56.12 MB, 60.56 MB for the server when $t=1, 2, 100$, respectively.

While being  collusion-resistant and fault-tolerant as well, \skeahel~is less efficient than \skeygc~and \skeahed. The computation of the KRE of 100 values takes to each client 1024 seconds (for $t=1$), 1749 seconds (for $t=2$).
The communication is 0.56 MB, 1.11 MB for each client and 111.37 MB, 222.67 MB for the server when $t=1, 2$, respectively. For $t=100$, our testbed ran out of memory. 

\begin{table}[tbp]
\centering
\scriptsize
\begin{tabular}{| l | r | r | r | r | r | r | r |}
\hline
\multicolumn{1}{ |l|}{} 		                & \multicolumn{1}{|c|}{\skeygc} & \multicolumn{3}{|c|}{\skeahed}  & \multicolumn{2}{|c|}{\skeahel}  \\ \hline
\multicolumn{1}{ |l|}{$t$}         & n/a     & 1 & 2 & 100 & 1 & 2 \\ \hline
\multicolumn{1}{ |l|}{Time (s)}    & 197.00 &  353.00  &   336.00 & 441.00 & 1024.00 & 1749.00\\
\multicolumn{1}{ |l|}{C-Bits (MB)} & 0.31 &  0.30   & 0.30 & 0.32 & 0.56 & 1.11\\
\multicolumn{1}{ |l|}{S-Bits (MB)} &  5.42 &   56.07  & 56.12 & 60.56 & 111.37 & 222.67\\
\hline
\end{tabular}
\caption{ \scriptsize{Performance Comparison for 100 clients: C-Bits (resp. S-Bits) denotes the number of bits sent by each client (resp. the server). $t$ is the secret sharing threshold, i.e., the number of clients that must contribute to the treshold decryption.}}
\label{Comparision_YGC_DGK_LTT}
\end{table}

We evaluated Algorithm \ref{FHE_KRE} of \skeshe~ at security level at least 110.
The result is illustrated in Figure \ref{Graph_SVR_SKE_SHE_Computation_Cost} for inputs with bitlength $\inputlen=16$. The computation is dominated by the inputs' comparison and takes less than one hour for 25 clients. We also evaluated in Figure \ref{Graph_SKE_SHE_THDEC_Cost} the performance of the threshold decryption with a $n$-out-of-$n$ secret sharing. For up to 40 clients threshold decryption costs less than 0.15 second. \skeshe~is practically less efficient than all other schemes, but has the best asymptotic complexity.

As a result \skeygc~is suitable for a setting where the server is non-colluding and clients cannot fail. If collusion and failure are an issue, then either \skeahel~or \skeahed~ or even \skeshe~is suitable. \skeahel~can be more time efficient than \skeahed~for up to 30 clients and a highly parallelizable server. \skeshe~ has the best asymptotic complexity, however, it requires more efficient somewhat homomorphic encryption schemes. 

\section{Conclusion}
\label{Conclusion}
In this paper we considered the problem of computing the KRE (with applications to benchmarking) of $n$ clients' private inputs using a server.   
The general idea of our solution is to sort the inputs, compute the rank of each input, use it to compute the KRE. The computation is supported by the server which coordinates the protocol and undertakes as much computations as possible. We proposed and compare different approaches based on garbled circuits or threshold HE. The server is oblivious, and does not learn the input of the clients. We also implemented and evaluated our schemes. 
As a result \skeygc~is suitable for a setting where the server is non-colluding and clients cannot fail. If collusion and failure are an issue, then either \skeahed~ or \skeshe~is suitable. \skeshe~ has the best asymptotic complexity, however, it requires more efficient somewhat homomorphic encryption schemes.


\bibliographystyle{splncs04}
\bibliography{SecureMedianEPrint}

\appendix

\section{ElGamal Encryption}
\label{sec:EcElGamal}
The threshold decryption in \skeahel, \skeahed~has be implemented using elliptic curve ElGamal (ECE) \cite{Koblitz1987}.
We briefly present ECE and its threshold decryption \cite{BonehShoup.2017}.

Let $\mathbb{G}$ be an elliptic curve group generated by a point $P$ of prime order $p$. The key generation
chooses $s \stackrel{\$}{\leftarrow} \mathbb{Z}_p$ and outputs $sk = s$ and $pk = s\cdot P$ as private and public key.
To encrypt an integer $m$, one chooses $r \stackrel{\$}{\leftarrow} \mathbb{Z}_p$ and outputs the ciphertext $c = (r\cdot P, m\cdot P + r\cdot pk)$.
To decrypt a ciphertext $c = (\ctxtcmt_1, \ctxtcmt_2)$, one computes $Q = \ctxtcmt_2 - \ctxtcmt_1\cdot sk$ and solves the discrete logarithm on $\mathbb{G}$.

\noindent Let $n, t$ be integers such that $t \leq n$. To support $t$-out-of-$n$ threshold decryption the secret key $sk=s$ is secret-shared using Shamir secret sharing scheme \cite{Shamir79} by choosing a random polynomial: 
$f(x) = s + \sum_{i = 1}^{t-1} a_{i}x^{i}$
and computing secret shares $s_i = f(i), 1\leq i \leq n$. 
Let $L_i(x)$ and $l_i$ be defined as:
$L_i(x) = \prod_{j=1, j\neq i}^{t} \frac{x-j}{i-j} ~\mbox{and}~ l_i = L_i(0) = \prod_{j=1, j\neq i}^{t} \frac{-j}{i-j}.$
Given $t$ shares $(1, s_1), \ldots, (t, s_t)$, the polynomial $g(x) = \sum_{i=1}^{t}L_i(x)\cdot s_i$ is the same as $f(x)$ since both have degree at most $t-1$ and match at $t$ points. Therefore $s = f(0) = g(0) = \sum_{i=1}^{t}s_i \cdot l_i$. The numbers $l_i$ are called \emph{Lagrange coefficients}.
The threshold key generation outputs secret key shares $\skshare_i = (s_i, l_i), 1\leq i \leq n$.
Let $\mathbb{I}_t \subseteq \{1, \ldots, n\}$ be a subset of $t$ clients and assume for simplicity $\mathbb{I}_t = \{1, \ldots, t\}$. To decrypt a ciphertext $c = (\ctxtcmt_1, \ctxtcmt_2)$ each client $C_i, i\in \mathbb{I}_t$ computes $m_i = \ctxtcmt_1 \cdot s_i \cdot l_i$. 
Then the combiner receives all $m_i$, computes $\ctxtcmt_2 - \sum_{i=1}^{t}m_i = \ctxtcmt_2 - \ctxtcmt_1 \cdot  \sum_{i=1}^{t} s_i \cdot l_i = \ctxtcmt_2 - \ctxtcmt_1 \cdot s = Q$ and solve the discrete logarithm on $\mathbb{G}$. 
This requires $O(1)$ to each client $C_i, i \in \mathbb{I}_t$ and $O(\log t)$ asymmetric operations to the combiner.

\section{Security Proofs}
\label{security_proofs}
Let the inherent leakage be $\mathcal{L} = \{k, n, t, \ctxtlen, \secprm, \inputlen\}$, i.e., protocol's parameters.

\begin{theorem}
\label{SKE_YGC_Security}
If the server $S$ is non-colluding and the AHE scheme is IND-CPA secure, then
\skeygc~1-privately computes $\func{F}{\kre}$ in the semi-honest model with leakage $\mathcal{L}_S = \mathcal{L}_i = \mathcal{L}$. Hence, there are simulators $\simulator_{C_i}$ for each $C_i$ and~ $\simulator_S$ for $S$ such that:
$$\simulator_{S}(\emptyset, \mathcal{L}_S) \stackrel{c}{\equiv} \view{S}{}(x_1, \ldots, x_n)~\mbox{and}$$
$$\simulator_{C_i}(x_{i}, \func{F}{\kre}(x_1, \ldots, x_n), \mathcal{L}_i) \stackrel{c}{\equiv} \view{C_i}{}(x_1, \ldots, x_n).$$

\end{theorem} 
\begin{proof}[Proof (Sketch)]
The leakage is clear as parties see only random strings. $\view{C_i}{}$ consists of: 
\begin{align}
(F^{ij}_{>}, (\bar{b}_{i}^{ij}, \bar{x}_i^{ij}), b'_{ij}, 
\lsem b^{ji}_j \oplus b_{ij} \rsem_i, \lsem b^{ji}_i \oplus b_{ij} \rsem_j, \lsem b_{ij} \rsem_j)_{1\leq j \leq n (i\neq j)}, 
\lsem \beta_i \rsem_i, \beta_i. \nonumber 
\end{align}
For each $m \in \view{C_i}{}$, $\simulator_{C_i}$ chooses random bit strings of length $|m|$.
The view of the server consists of:
\begin{align}
\langle F^{ij}_{>}, (\bar{b}_{i}^{ij}, \bar{x}_i^{ij}), (\bar{b}_{j}^{ij}, \bar{x}_j^{ij}), b'_{ij}, 
\lsem b^{ji}_j \oplus b_{ij} \rsem_i, \lsem b^{ji}_i \oplus b_{ij} \rsem_j,  
\lsem b^{ij}_i \oplus b_{ij} \rsem_j, \lsem b^{ij}_j \oplus b_{ij} \rsem_i, \nonumber \\ 
\lsem b_{ij} \rsem_j, \lsem b_{ij} \rsem_i \rangle_{\paired(i,j)=true}, 
\langle \lsem r_i \rsem_i, \lsem \beta_i \rsem_i \rangle_{1\leq i \leq n}, \langle \lsem m_i \rsem_j
\rangle_{1 \leq i,j \leq n}. \nonumber
\end{align}
For each $m \in \view{S}{}$, $\simulator_{S}$ chooses random bit strings of length $|m|$.
\end{proof}

\begin{theorem}
\label{SKE_AHE1_Security}
Let $t \in \mathbb{N}$ and $\tau < t$. If the server $S$ is non-colluding and the AHE scheme is IND-CPA secure, then  \skeahel~and \skeahed~$\tau$-privately compute $\func{F}{\kre}$ in the semi-honest model with leakage $\mathcal{L}_S = \mathcal{L}_i = \mathcal{L}$. Hence,  let $I = \{i_{1}, \ldots, i_{\tau}\}$, $\mathcal{L}_{I} = \bigcup_{i\in I} \mathcal{L}_i$, there exists a simulator $\simulator_{I}$ such that:
$$\simulator_{I}((x_{i_1}, \ldots, x_{i_{\tau}}), \func{F}{\kre}(x_1, \ldots, x_n), \mathcal{L}_{I}) \stackrel{c}{\equiv} \view{I}{}(x_1, \ldots, x_n).$$
\end{theorem} 
\begin{proof}[Proof (Sketch)]
The leakage is clear as parties see only random strings (IND-CPA ciphertexts, random shares or random bits). 

In \skeahel~, all messages can be simulated by choosing random bit strings of the corresponding length. 
However, the simulation of Step \ref{Server_forward_Th_Dec_AHE1} must be coherent with Step \ref{Client_Rec_Comp_Res_AHE1}. Each client receives random shares in Step \ref{Server_forward_Th_Dec_AHE1}, runs the final decryption $\findec(.)$ in Step \ref{Client_Rec_Comp_Res_AHE1} and learns a random bit. Let $C_l$ be a client with $l \in I$. To simulate Steps \ref{Server_forward_Th_Dec_AHE1} and \ref{Client_Rec_Comp_Res_AHE1}, the simulator chooses $t$ random values for Step \ref{Server_forward_Th_Dec_AHE1} such that running $\findec(.)$ returns the random bit simulated in Step \ref{Client_Rec_Comp_Res_AHE1}.

For example, if the underlying AHE is ECC ElGamal (ECE), then a ciphertext has the form $c = (\ctxtcmt_1, \ctxtcmt_2) = (r\cdot P, m\cdot P + r\cdot pk)$.
For each ECE ciphertext $c = (\ctxtcmt_1, \ctxtcmt_2) = (r\cdot P, m\cdot P + r\cdot pk)$ that must be decrypted in Step \ref{Client_Rec_Comp_Res_AHE1}, $C_l$ gets $\ctxtcmt_2$ and $t$ partial decryption results $\ctxtcmt_{11}, \ldots, \ctxtcmt_{1t}$ of $\ctxtcmt_1$ in Step \ref{Server_forward_Th_Dec_AHE1}. To simulate this,  the simulator chooses a random bit $b$ and a random $\tilde{\ctxtcmt}_2$. Then it computes $\tilde{\ctxtcmt}_1 = \tilde{\ctxtcmt}_2 - b \cdot P$ and generates random $\tilde{\ctxtcmt}_{11}, \ldots, \tilde{\ctxtcmt}_{1t}$ such that $\sum_{i=1}^{t} \tilde{\ctxtcmt}_{1i} = \tilde{\ctxtcmt}_1$ in $\mathbb{G}$.  

The proof for \skeahed~ is similar. 
Ciphertexts and random shares are simulated with equally long random strings and Steps \ref{Forward_to_combiner} and \ref{combiner_Select_KRE} in \skeahed~ are simulated as above for Steps \ref{Server_forward_Th_Dec_AHE1} and \ref{Client_Rec_Comp_Res_AHE1} in \skeahel.
\end{proof}

\begin{theorem}
\label{SKE_SHE_Security}
Let $t \in \mathbb{N}$ and $\tau < t$. If the server $S$ is non-colluding and the SHE scheme is IND-CPA secure, then  \skeshe~$\tau$-privately computes $\func{F}{\kre}$ in the semi-honest model with leakage $\mathcal{L}_S = \mathcal{L}_i = \mathcal{L}$. Hence,  let $I = \{i_{1}, \ldots, i_{\tau}\}$ denote the indexes of corrupt clients, $\mathcal{L}_{I} = \bigcup_{i\in I} \mathcal{L}_i$ denote their joint leakages and $\view{I}{}(x_1, \ldots, x_n)$ denote their joint views, there exists a simulator $\simulator_{I}$ such that:
$$\simulator_{I}((x_{i_1}, \ldots, x_{i_{\tau}}), \func{F}{\kre}(x_1, \ldots, x_n), \mathcal{L}_{I}) \stackrel{c}{\equiv} \view{I}{}(x_1, \ldots, x_n).$$
\end{theorem} 
\begin{proof}
The leakage is clear as parties see only random strings (IND-CPA ciphertexts or partial decryption results). 
The security is also straightforward as the computation is almost completely done by the server alone and encrypted under an IND-CPA encryption. Moreover, the partial decryption reveals only partial result to each decryptor. 
\end{proof}

Recall that our adversary is semi-honest. 
In \skeygc, a server collusion reveals all inputs to the adversary. In \skeahel~and \skeahed, a server collusion only increase the leakage as long as the number of corrupted clients is smaller than $t$. For example in \skeahel, the adversary can learn the order of the inputs whose comparison bits are final decrypted by a corrupted client in Step \ref{Client_Rec_Comp_Res_AHE1}. In \skeshe, the $\kre$ is homomorphically computed by the server such that the clients are only required for the decryption of one ciphertext encrypting the $\kre$. Moreover, the ciphertexts are encrypted using the threshold public key. As a result, assuming semi-honest adversary and a collusion set containing less than $t$ clients, a server collusion leaks no more information than $k, n, t, \ctxtlen, \secprm, \inputlen$.

\section{Complexity analysis}
\label{complexity_analysis}
In this section, we discuss the complexity of our schemes. 
We will use $\ctxtlen$ and $\secprm$ as length of asymmetric ciphertext and symmetric security parameter. 

\subsection{\skeygc~Protocol}
A GC for the comparison of two $\inputlen$-bit integers consists of $\inputlen$ AND-gates resulting in $4\inputlen$ symmetric ciphertexts \cite{KolesnikovS.2008, KolesnikovSS.2009}. It can be reduced by a factor of 2 using the \textit{halfGate} optimization \cite{ZahurRE.2015} at the cost of performing two cheap symmetric operations (instead of one) during GC evaluation.

We do the analysis for the case where $n$ is odd (the even case is similar). From Lemma \ref{Paired_CMP_Count_By_CLT}, each client generates $(n-1)/2$ GCs resulting in $(n-1)\inputlen$ symmetric operations. 
The computation of encrypted comparison bits (Steps \ref{ygc_Comparison_bit} to \ref{ygc_Comparison_bit1}) and the computation of the KRE's ciphertext require $O(n)$ asymmetric operations to each client. Finally, each client has to decrypt one ciphertext in Step \ref{ygc_revealing_KRE}. As a result, the computation complexity of each client is therefore 
$O((n-1)\inputlen)$ symmetric and $O(2n+1)$ asymmetric operations. 
In communication, this results in $n\ctxtlen$ bits for the asymmetric ciphertexts, $2\inputlen\secprm(n-1)/2$ bits for the GCs and $\inputlen\secprm(n-1)/2$ for the garbled inputs and $n\ctxtlen$ bits for handling the server's leakage. In total each client sends 
$2n\ctxtlen + \frac{3\inputlen\secprm(n-1)}{2}$.  

The server evaluates $n(n-1)/2$ GCs each consisting of $2\inputlen$ symmetric ciphertexts. Computing the rank (Steps \ref{ygc_compute_rank} to \ref{ygc_compute_rank1}) requires $O(n\log n + n)$ operations to the server. Finally, the server evaluates $\log n + n$ asymmetric operations to compute the KRE ciphertext for each client (Steps \ref{ygc_revealing_KRE} to \ref{ygc_revealing_KRE1}). The total computation complexity of the server is 
$O(n(n-1)\inputlen)$ symmetric and $O((n+1)\log n + 2n)$. 
In communication, the server sends $n(n-1)$ asymmetric ciphertexts in Steps \ref{ygc_Comparison_bit} to \ref{ygc_Comparison_bit1}, $n$ asymmetric ciphertexts in Steps \ref{ygc_compute_rank} to \ref{ygc_compute_rank1} and $n$ asymmetric ciphertexts in Steps \ref{ygc_revealing_KRE} to \ref{ygc_revealing_KRE1}. This results in a total of
$(n^2+n)\ctxtlen$ bits.

\subsection{\skeahel~Protocol}
Each client performs $O(\inputlen)$ operations in Step \ref{ciphertext_upload_AHE1}, $O(n \inputlen t)$ operations in Step \ref{Client_Th_Dec_AHE1}, $O(n \inputlen t)$ operations in Step \ref{Client_Dec_Comp_Res_AHE1}, $O(\log t)$ operations in Step \ref{Client_Rec_Comp_Res_AHE1}, $O(1)$ operations in Step \ref{Client_Send_KRE_AHE1}, eventually $O(1)$ and $O(n)$ operations in Steps \ref{Client_Part_Dec_KRE_AHE1} and \ref{Client_Send_Part_Dec_KRE_AHE1}, and $O(\log t)$ operations in Step \ref{Client_Rec_KRE_AHE1}. This results in a total of 
$O(\inputlen + 2n\inputlen t + 2\log t + n + 1)$ 
asymmetric operations.

Each client sends $(2\inputlen+1)\ctxtlen$ bits in Step \ref{ciphertext_upload_AHE1}, $n\inputlen t\ctxtlen$ bits in Step \ref{Client_Th_Dec_AHE1}, $\ctxtlen$ bits in Step \ref{Client_Send_KRE_AHE1}, eventually $n\ctxtlen$ bits in Step \ref{Client_Send_Part_Dec_KRE_AHE1}. This results in a total of 
$(2\inputlen + n\inputlen t + n + 2)\ctxtlen$
bits for each client.

The main cryptographic operations of server happen in the evaluation of the Lin-Tzeng protocol in Step \ref{HOM_Compare_Call_AHE1}. The comparison of two values takes $2\inputlen$ asymmetric operations. 
As a result the server performs 
$O(2\inputlen n^2)$ 
asymmetric operations for all comparisons.  

The server sends $n^2 \inputlen t\ctxtlen$ bits in Step \ref{TH_Decryption_Request_Call_AHE1} and $(n^2t+1)\ctxtlen$ bits in Step \ref{Server_forward_Th_Dec_AHE1}, $t\ctxtlen$ bits in Step \ref{Server_Send_KRE_Encrypted_AHE1} and $nt\ctxtlen$ bits in Step \ref{Server_Reveal_KRE_AHE1}. This results in a total of 
$(n^2\inputlen t + n^2t + nt + t + 1)\ctxtlen$ 
bits for the server.
 
\subsection{\skeahed~Protocol}

Since \skeahed~also requires the predicate \paired~  as \skeygc, we do the analysis for the case where $n$ is odd (the even case is similar).

Each client performs $O(\inputlen+1)$ operations in Step \ref{ciphertext_upload_AHE2}, 
$O(\frac{7 \inputlen (n-1)}{2})$ operations in Step \ref{DGK_Compare_Call}, 
$O(t)$ operations in Step \ref{Client_Th_Dec_AHE2} and $O(\log t)$ in Step \ref{Combiner_Rec_KRE_AHE2}, $O(n)$ operations in Step \ref{combiner_Send_KRE_Ciphertext_AHE2} and $O(1)$ operations in Step \ref{Client_Dec_KRE_AHE2}. This results in a total of 
$O(\inputlen + \frac{7 \inputlen (n-1)}{2} + t + \log t + n + 1)$ 
asymmetric operations.

Each client sends $(\inputlen+1)\ctxtlen$ bits in Step \ref{ciphertext_upload_AHE2}, $\frac{\ctxtlen(n-1)}{2}$ bits (when the client is head) and $\frac{(\inputlen+1)\ctxtlen(n-1)}{2}$ (when the client is tail) in Step \ref{DGK_Compare_Call}, $t\ctxtlen$ bits in Step \ref{Client_Th_Dec_AHE2} and $n\ctxtlen$ bits in Step \ref{combiner_Send_KRE_Ciphertext_AHE2}. This results in a total of 
$(\inputlen\frac{(n+1)}{2} + 2n + t)\ctxtlen$
bits for each client.

The cryptographic operations of the server happen in \textsc{ComputeKreAhe} (Algorithm \ref{Compute_Rank}) that is called in  Step \ref{Call_Compute_KRE} of Protocol \ref{AHE2_KRE_Protocol}. The server performs 
$O(n^2+n)$ 
asymmetric operations.

The server sends $\frac{(\inputlen\ctxtlen + (\inputlen+1)\ctxtlen)n(n-1)}{2}$ bits in Step \ref{DGK_Compare_Call}, $nt\ctxtlen$ bits in Steps \ref{Decryption_Request_Call_AHE2} and \ref{Forward_to_combiner}, $n\ctxtlen$ bits in Step \ref{Server_Reveal_KRE_AHE2}. This results in a total of 
$(\frac{(2\inputlen+1)n(n-1)}{2} + 2nt + n)\ctxtlen$ 
bits for the server.

\subsection{\skeshe~Protocol}
Each client has $O(\inputlen)$
computation cost ($\inputlen + 1$ encryptions in Step \ref{ciphertext_upload_TH_FHE} and eventually one partial decryption in Step \ref{Client_Send_KRE_Share_TH_FHE}) and a communication cost of 
$(\inputlen+n+1)\ctxtlen$ bits.

The cryptographic operations of the server happen in \textsc{ComputeKreShe} (Algorithm \ref{FHE_KRE}) that is called in  Step \ref{Server_Compute_KRE_TH_FHE} of Protocol \ref{FHE_KRE_Protocol}. The SHE comparison circuit has depth $\log (\inputlen-1)+1$ and requires $O(\inputlen\log\inputlen)$ homomorphic multiplications \cite{CheonKL15, CheonKK16}. For all comparisons the server performs, therefore, $O(n^2\inputlen\log\inputlen)$ multiplication. In Step \ref{compute_f_r} of Algorithm \ref{FHE_KRE}, the computation of $\lsem \prod_{j=1, j\neq k}^{n}(r_i - j) \rsem$ has depth $\log n$ and requires $O(n\log n)$ homomorphic multiplications. Step \ref{inner_product} of Algorithm \ref{FHE_KRE} adds an additional circuit depth and requires $O(n)$ homomorphic multiplications. As a result, Algorithm \ref{FHE_KRE}  has a total depth of $\log (\inputlen-1) + \log n + 2$ and requires 
$O(n^2\inputlen\log\inputlen + n\log n + n)$ 
homomorphic multiplications.

The server sends $t\ctxtlen$ bits in Step \ref{Server_Send_KRE_Encrypted_TH_FHE} and $nt\ctxtlen$ bits in Step \ref{Server_Reveal_KRE_TH_FHE} resulting in a total of 
$(t + nt)\ctxtlen$ 
bits.

\end{document}